\documentclass[11pt]{article}

\usepackage{graphicx}    
\usepackage[margin=1in]{geometry}
\usepackage{amsmath}
\usepackage{amsfonts}
\usepackage{hyperref}
\usepackage{subcaption}
\usepackage{amssymb}
\usepackage{amsthm}
\usepackage{array}
\usepackage{latexsym}
\usepackage{enumitem}
\usepackage{soul}
\usepackage{microtype}
\usepackage{makecell}

\newtheorem{theorem}{Theorem}[section]
\newtheorem{lemma}[theorem]{Lemma}

\newtheorem{corollary}[theorem]{Corollary}

\newtheorem{no-longer-open}{\st{Open} Problem}
\newtheorem{open}{Open Problem}

\theoremstyle{definition}
\newtheorem{definition}[theorem]{Definition}

\def\defn#1{\textit{\textbf{#1}}}

\let\realbfseries=\bfseries
\def\bfseries{\realbfseries\boldmath}

{\makeatletter
 \gdef\xxxmark{%
   \expandafter\ifx\csname @mpargs\endcsname\relax 
     \expandafter\ifx\csname @captype\endcsname\relax 
       \marginpar{xxx}
     \else
       xxx 
     \fi
   \else
     xxx 
   \fi}
 \gdef\xxx{\@ifnextchar[\xxx@lab\xxx@nolab}
 \long\gdef\xxx@lab[#1]#2{\textbf{[\xxxmark #2 ---{\sc #1}]}}
 \long\gdef\xxx@nolab#1{\textbf{[\xxxmark #1]}}
}

{\makeatletter
 \gdef\labelthis#1#2{{\edef\@currentlabel{#2}\label{#1}\@currentlabel}}}

\let\epsilon\varepsilon

\def\tfrac#1#2{{\textstyle\frac{#1}{#2}}}

\title{Finding Closed Quasigeodesics on Convex Polyhedra%
  \thanks{A preliminary version of this paper appeared at the
    \emph{36th International Symposium on Computational Geometry (SoCG 2020)}
    \cite{Quasigeodesics_SoCG2020}.}}

\author{%
  Erik D. Demaine%
    \thanks{Computer Science and Artificial Intelligence Laboratory,
      Massachusetts Institute of Technology, Cambridge, MA, USA,
      \protect\url{edemaine@mit.edu}}
\and
  Adam C. Hesterberg%
    \thanks{John A. Paulson School of Engineering and Applied Sciences,
      Harvard University, Cambridge, MA, USA,
      \protect\url{ahesterberg@seas.harvard.edu}}
\and
  Jason S. Ku%
    \thanks{Department of Mechanical Engineering,
      National University of Singapore, Singapore,
      \protect\url{jasonku@nus.edu.sg}}
}

\date{}

\begin{document}
\maketitle

\begin{abstract}

A \defn{closed quasigeodesic} is a closed curve on the surface of a polyhedron 
with at most $180^\circ$ of surface on both sides at all points;
such curves can be locally unfolded straight.
In 1949, Pogorelov proved that every convex polyhedron has at least three
(non-self-intersecting) closed quasigeodesics,
but the proof relies on a nonconstructive topological argument.
We present the first finite algorithm to find a closed quasigeodesic on a
given convex polyhedron, which is the first positive progress on a 1990
open problem by O'Rourke and Wyman.
The algorithm also establishes
a pseudopolynomial upper bound on the total number of visits to faces (number of line segments),
namely, $O\left(\frac{n \, L^2}{\epsilon^2 \, \ell^2}\right)$
where $n$ is the number of vertices of the polyhedron,
$\epsilon$ is the minimum curvature of a vertex,
$L$ is the length of the longest edge, and $\ell$ is the smallest
distance within a face between a vertex and a nonincident edge
(minimum feature size of any face).
On the real RAM, the algorithm's running time is also pseudopolynomial, namely 
$O\left(\frac{L^2}{\epsilon^2 \, \ell^2} \, n \lg n\right)$.
On a word RAM, the running time grows to
$O\left(
\tfrac{b^2 \, \Delta^{36} \, L^{146}}{\epsilon^{98} \, \ell^{146}} \, n \lg n
\cdot
2^{O( |R| )}\right)$,
where $\Delta \leq n$ is the polyhedron's maximum vertex degree,
assuming the polyhedron's intrinsic geometry is given by constant-size radical expressions
with $b$-bit integers and at most $|R|$ distinct square-roots.
Along the way, we introduce the expression RAM model of computation,
formalizing a connection between the real RAM and word RAM
hinted at by past work on exact geometric computation.

\end{abstract}

\section{Introduction}

A \defn{geodesic} on a surface is a path that is locally shortest at every
point,
i.e., cannot be made shorter by modifying the path in a small neighborhood.
A \defn{closed geodesic} on a surface is a \defn{loop} (closed curve)
with the same property;
notably, the locally shortest property must hold at all points,
including the ``wrap around'' point where the curve meets itself.
In 1905, Poincar\'e \cite{Poincare} conjectured that every convex surface
has a non-self-intersecting closed geodesic.%
\footnote{Non-self-intersecting (quasi)geodesics are often called \emph{simple}
  (quasi)geodesics in the literature; we avoid this term to avoid ambiguity
  with other notions of ``simple''.}
In 1927, Birkhoff \cite{OneGeodesic} proved this result,
even in higher dimensions (for any smooth metric on the $n$-sphere).
In 1929, Lyusternik and Schnirelmann~\cite{ThreeGeodesicsWrong} claimed that
every smooth surface of genus $0$
in fact has at least \emph{three} non-self-intersecting closed geodesics.
Their argument ``contains some gaps'' \cite{ThreeGeodesicsDiscuss},
filled in later by
Ballmann~\cite{ThreeGeodesics}.

For non-smooth surfaces (such as polyhedra), an analog of a geodesic is a
\defn{quasigeodesic} --- a path with $\leq 180^\circ$ of surface on both sides
locally at every point along the path.
Equivalently, a quasigeodesic can be locally unfolded to a straight line:
on a face, a quasigeodesic is a straight line; at an edge, a quasigeodesic is a
straight line after the faces meeting at that edge are unfolded (developed)
flat at that edge; and at a vertex of curvature $\kappa$
(that is, a vertex whose sum of incident face angles is $360^\circ - \kappa$), 
a quasigeodesic entering the vertex at a given angle can exit it anywhere
in an angular interval of length $\kappa$,
as in Figure~\ref{QuasigeodesicVertexFigure}.
Analogously, a \defn{closed quasigeodesic} is a loop which is quasigeodesic.
In 1949, Pogorelov \cite{DiscreteThreeGeodesics} proved that every convex
surface has at least three non-self-intersecting closed quasigeodesics, by applying the 
theory of geodesics on smooth surfaces
to smooth approximations of arbitrary convex surfaces and taking limits.

\begin{figure}
\centering
\includegraphics[width=0.75\linewidth]{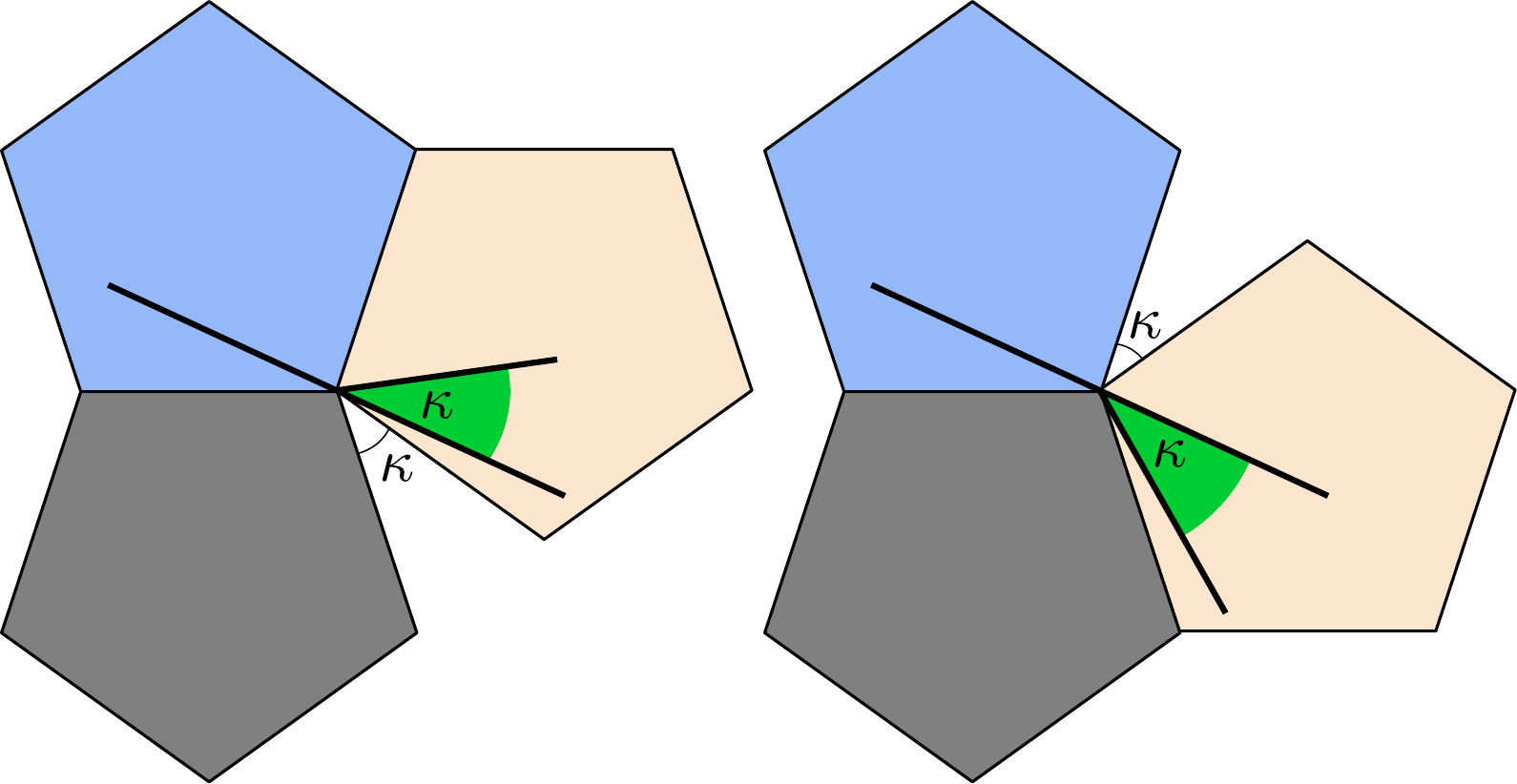}
\caption{At a vertex of curvature $\kappa$, there is a $\kappa$-size interval
of angles in which a segment of a quasigeodesic can be extended: the segment of
geodesic starting on the left can continue straight in either of the pictured
unfoldings, or any of the intermediate unfoldings in which the right pentagon
touches only at a vertex.}
\label{QuasigeodesicVertexFigure}
\end{figure}

The existence proof of three closed quasigeodesics is nonconstructive,
because the smooth argument uses a nonconstructive topological argument
(a homotopy version of the intermediate value theorem).%
\footnote{A proof sketch for the existence of one closed geodesic on a smooth
  convex surface is as follows.  By homotopy, there is a continuous deformation
  of a small clockwise loop into its (counterclockwise) reversal that avoids
  self-intersection throughout.  Consider the deformation that minimizes
  the maximum arclength of any loop during the deformation.  By local
  cut-and-paste arguments, the maximum-arclength intermediate loop is in fact
  a non-self-intersecting closed geodesic.
  The same argument can be made for the nonsmooth case.}
In 1990, Joseph O'Rourke and Stacia Wyman posed the problem of finding a
polynomial-time algorithm to find any closed quasigeodesic on a given convex
polyhedron
(aiming in particular for a non-self-intersecting closed quasigeodesic)
\cite{ORourke-2020}.
This open problem was stated during the open
problem session at SoCG 2002 (by O'Rourke) and finally appeared in print
in 2007
\cite[Open Problem 24.24]{Demaine-O'Rourke-2007}.
Two negative results described in
\cite[Section~24.4]{Demaine-O'Rourke-2007}
are that an $n$-vertex polyhedron can have $2^{\Omega(n)}$
non-self-intersecting closed quasigeodesics
\cite[Theorem 24.4.1]{Demaine-O'Rourke-2007}
(a previously unpublished result by Aronov and O'Rourke) and that, for any $k$,
there is a convex polyhedron whose shortest closed geodesic is not composed of
$k$ shortest paths (an unpublished result from the discussion at SoCG 2002).

Even a finite algorithm is not obvious.
One tempting approach is to find two
(or more)
shortest paths whose union
is a closed quasigeodesic. For example, the source unfolding
\cite{Agarwal-Aronov-O'Rourke-Schevon-1997,Demaine-O'Rourke-2007}
(Voronoi diagram) from a polyhedron vertex $V$ consists of all points on
the polyhedron having multiple shortest paths to~$V$,
as in Figure~\ref{fig:source}.
Can we find a point $P$ on the source unfolding and two shortest paths between
$P$ and $V$ whose union forms a closed quasigeodesic?
We know that there is a closed quasigeodesic through some vertex,
because if we have a closed quasigeodesic through no vertices,
we can shift it on the surface (remaining parallel to itself)
until it hits at least one vertex.
By guessing, we can assume that $V$ is a vertex visited by
a closed quasigeodesic.
But there might not be any choice for $P$ that makes the two shortest paths
meet with sufficiently straight angles at both $V$ and~$P$,
as in Figure~\ref{fig:source}.

\begin{figure}
\centering
\includegraphics[width=0.9\linewidth]{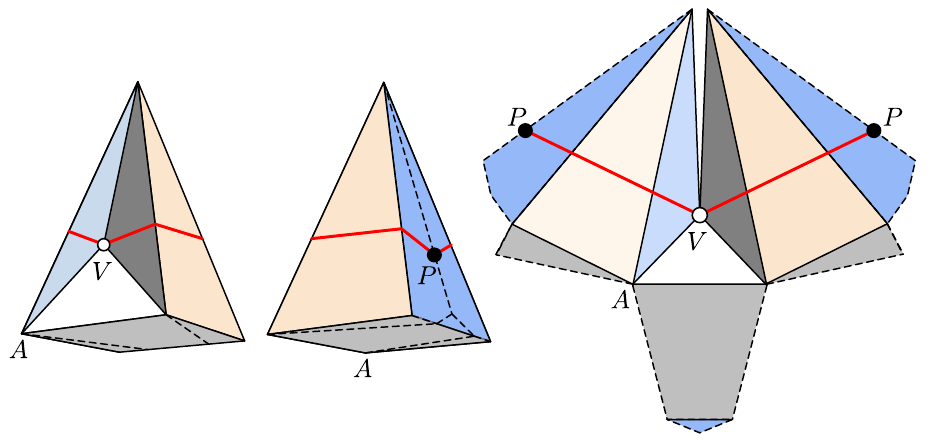}
\caption{A source unfolding from vertex $V$ of a six-vertex polyhedron
  (the convex hull of a square-based pyramid and vertex $V$ which is slightly
  outside the pyramid),
  similar to \cite[Figure~1]{Agarwal-Aronov-O'Rourke-Schevon-1997}
  and \cite[Figure~24.2]{Demaine-O'Rourke-2007}.
  No closed quasigeodesic can be formed by two shortest paths from $V$ to
  another point $P$, but there is a (vertical) closed quasigeodesic passing through~$V$.}
\label{fig:source}
\end{figure}

A more general approach is to
argue that there is a closed quasigeodesic consisting of some function $s(n)$
(e.g., $O(n)$) segments on faces.
If true, there are $O(n)^{s(n)}$ combinatorial types
of quasigeodesics to consider, and each can be checked via the existential
theory of the reals (in exponential time), resulting in an exponential-time
algorithm.
But we do not know any such bound $s(n)$.
It seems plausible that the ``short'' closed quasigeodesics from
the nonconstructive proofs satisfy $s(n) = O(n)$, say,
but as far as we know the only proved property about them
is that they are non-self-intersecting, which does
not suffice: a quasigeodesic can wind many times,
turn around, and symmetrically unwind, all without
collisions, as in Figure~\ref{fig:prism_many}.
Polyhedra such as isosceles tetrahedra
have arbitrarily long non-self-intersecting closed geodesics
(and even infinitely long non-self-intersecting geodesics)
\cite{Itoh-Rouyer-Vilcu-2019}, so the only hope is to find an upper bound
$s(n)$ on some (fewest-edge) closed quasigeodesic.

\begin{figure}
\centering
\includegraphics[width=0.9\linewidth]{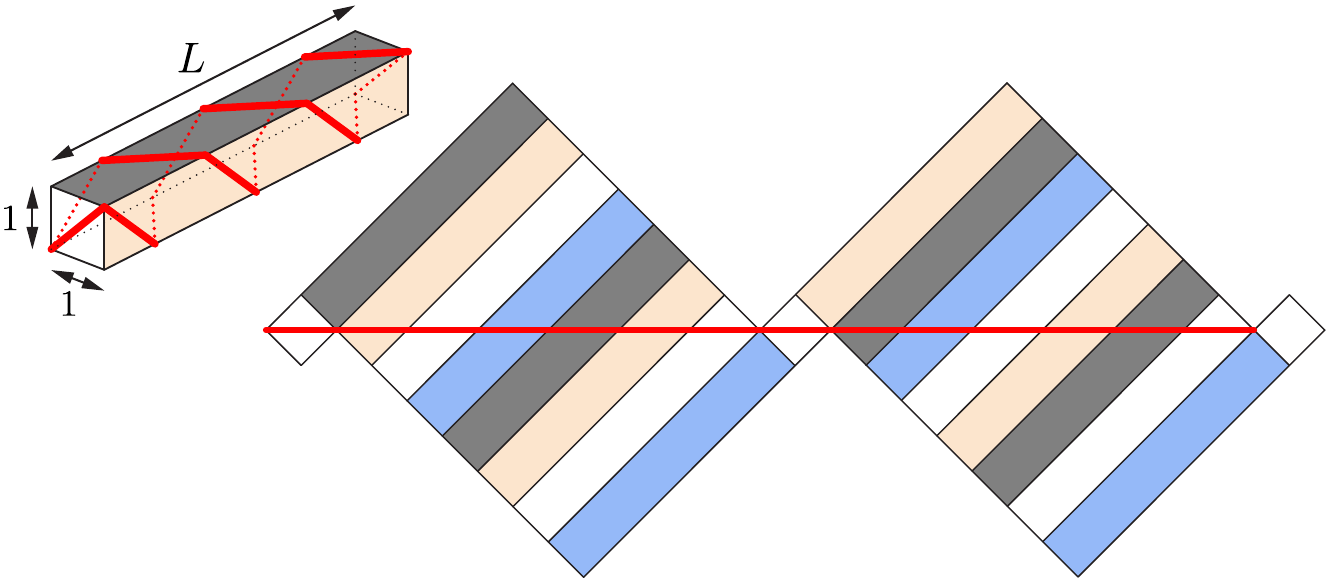}
\caption{
Non-self-intersecting quasigeodesics may cross a face many times. For example, a
$1\times 1\times L$ rectangular prism admits closed quasigeodesics which cross a
face $\Omega(L)$ times.}
\label{fig:prism_many}
\end{figure}

Since the conference version of this paper \cite{Quasigeodesics_SoCG2020},
two other approaches have been developed.
Sharp and Crane \cite{Sharp-Crane-2020} gave a practical heuristic
based on edge flipping, but it has no proof of convergence.
\label{sec:simple-geodesics-time}
Chartier and de Mesmay \cite{simple-geodesics}
gave a finite guaranteed-correct real-RAM algorithm to find a
non-self-intersecting closed quasigeodesic,
but the running time is pseudo-exponential
(exponential in both $n$ and $L/\ell$ defined below).
Their algorithm is a version of the idea above.
They prove a pseudopolynomial upper bound on $s(n)$,
generalized to four parameters:
in our notation, $s(n,\Delta,L,\ell) = O(\Delta n L/\ell)$, where
$\Delta$ is the maximum vertex degree,
$L$ is the length of the longest edge, and $\ell$ is the smallest
distance within a face between a vertex and a nonincident edge
(minimum feature size of any face).
Their bound holds even when
restricted to non-self-intersecting closed quasigeodesics.
The resulting running time of their algorithm is
$n^{O(\Delta n L/\ell)}$ on the real RAM,
by brute forcing over the possible sequences.
(They avoid the existential theory of the reals by sliding any
non-self-intersecting closed quasigeodesic to hit at least one vertex,
and then drawing straight lines.)

\subsection{Our Results}

We develop in Section~\ref{RealRAMAlgorithmSection} the first finite algorithm
that finds a
%
%
closed quasigeodesic on a given convex polyhedron, using 
$O\left(\frac{L^2}{\epsilon^2 \, \ell^2} \, n \lg n\right)$
real-RAM operations (arithmetic and square roots),
where $n$ is the number of vertices of the polyhedron,
$\epsilon$ is the smallest curvature at a vertex,
$L$ is the length of the longest edge, and $\ell$ is the smallest distance
within a face between a vertex and a nonincident edge
(minimum feature size of any face).%
\footnote{As usual, $\lg n$ denotes $\log_2 n$.}
Furthermore, the found closed quasigeodesic consists of
$O\left(\frac{n \, L^2}{\epsilon^2 \, \ell^2}\right)$
segments on faces, which is a new finite upper bound
$s(n,\epsilon,L,\ell)$.%
\footnote{The conference version of this paper \cite{Quasigeodesics_SoCG2020}
  contained the real-RAM algorithm presented here, but did not yet prove a
  bound on $s(n,\epsilon,L,\ell)$.
  Thus, chronologically, our real-RAM algorithm was first,
  while Chartier and de Mesmay's bound on $s(\cdots)$ \cite{simple-geodesics}
  was first.
}
These time and segment bounds are both pseudopolynomial.
In particular, this running time is much faster than the
pseudo-exponential running time of \cite{simple-geodesics}.
(The two bounds on $s$ are both pseudopolynomial, but their details are
different enough that we cannot apply either bound to the other algorithm.)

The real-RAM model of computation is common in computational geometry,
but it is an unrealistic model for digital computers
which are restricted to finite-precision computation.
We introduce in Section~\ref{sec:model} a model of computation
for realistic manipulation of radical real numbers,
called the \defn{expression RAM}.
This model provides a simple interface for manipulation of radical expressions
(arithmetic, roots, and comparisons), similar to the real RAM.
The key feature is that the expression RAM model can be implemented on top of
either the real RAM or the word RAM (the standard model for digital computers),
with different operation costs.
In particular, we give the first general transformation of a real-RAM algorithm
into a word-RAM algorithm with ``only'' singly exponential slowdown.
Furthermore,
when an expression has height $O(\lg n)$ and
contains at most $|R|$ distinct $O(1)$-roots
(e.g., square roots) and integers of at most $b$ bits,
the transformation has slowdown $b \cdot n^{O(1)} \cdot 2^{O(|R|)}$.
For example, real-RAM expressions of constant size can be implemented
with slowdown $O(b)$, preserving polynomial time bounds.%
\footnote{An earlier version of this paper \cite{Quasigeodesics_SoCG2020}
  mistakenly claimed that our closed-quasigeodesic algorithm fit within this
  ``$O(1)$-expression RAM'', and thus mistakenly claimed that the running time
  remained pseudopolynomial on the word RAM.}
These results formalize and analyze
the existing practical work on exact geometric computation
pioneered by LEDA/CGAL reals \cite{SepBound} and CORE reals \cite{CORE}.

We develop in Section~\ref{ExpressionRAMAlgorithmSection}
an expression-RAM version of our algorithm for finding a closed quasigeodesic,
whose running time on a word RAM is
$O\left(
\tfrac{b^2 \, \Delta^{36} \, L^{146}}{\epsilon^{98} \, \ell^{146}} \, n \lg n
\cdot
2^{O( |R| )}\right)$,
where $\Delta$ is the maximum vertex degree,
and the polyhedron is specified intrinsically by a 2D coordinatization
of each face, where each coordinate is a constant-size radical expression
with $b$-bit integers and at most $|R|$ distinct square-roots.
When $|R| = O(\lg n)$ or
$|R| = O\left(\lg \left({n \, L \over \epsilon \, \ell}\right)\right)$,
this running time is pseudopolynomial.
More generally, the represented gluing of faces only needs to be an
Alexandrov gluing, which by Alexandrov's Theorem corresponds to a
unique convex polyhedron \cite{Alexandrov-1996,Demaine-O'Rourke-2007}
(see Section~\ref{Polyhedral Inputs}).

\label{sec:sum-of-square-roots}
In fact, it is unlikely that a closed quasigeodesic can be found
in worst-case subexponential time on the word RAM, until we resolve the
famous sum-of-square-roots open problem \cite{TOPP-33}.
For example, comparing the lengths of two paths given by $O(n)$ segments
on faces (as required by any shortest-path algorithm)
requires comparing two sums of $O(n)$ square roots,
which is not known to be solvable in subexponential time.
(For the same reason, the natural decision versions of Euclidean shortest paths
and minimum spanning tree are not known to be in NP.)
For closed quasigeodesics,
it seems necessary to answer decision problems such as
``is there a geodesic path connecting two given vertices that visits the
given sequence of edges/faces in order?''
While we do not know a reduction from sum-of-square-roots to this problem,
the most natural solution by unfolding the faces one after the other
involves arithmetic (sums and multiplications) over the lengths of the edges,
which are square roots.

Our algorithm supports more general inputs than the representation described
above, but the time bound gets more complicated.
For example, given a \emph{triangulated} polyhedron via
\emph{3D vertex coordinates}, each given by a constant-size radical expression
with $b$-bit integers and at most $|R|$ distinct square-roots,
and having $|\Lambda|$ edge lengths,
we obtain a running time of
$O\left(\tfrac{b^2 \, \Delta^{36} L^{146}}{\epsilon^{98} \, \ell^{146}}
\, n \lg n \cdot
2^{O(|\Lambda|^3+|R|)}\right)$
on the word RAM; see Section~\ref{Polyhedral Inputs}.

We can also apply our techniques to Chartier and de Mesmay's algorithm for
finding a \emph{non-self-intersecting} closed quasigeodesic
\cite{simple-geodesics}.
The geometric part of their algorithm lays out the faces met by a hypothetical
geodesic in the plane, and checks that the straight line properly meets those
faces, which is a step we analyze in
Lemma~\ref{lem2:ray follow exact}.
By Equation~\eqref{eq:T'(M,h)},
the running time for this step on the word RAM is
$
    O\left(
      b^2 \cdot \big( \frac{n L}{\ell}\big)^{O(1)} \cdot 2^{O(|R|)}
    \right)
$.
The algorithm runs this step for every possible
unfolding sequence of length $O(\Delta n L/\ell)$,
multiplying the running time by $n^{O(\Delta n L/\ell)}$.
Thus the total running time on the word RAM is
$
    O\left(
      b^2 \cdot 2^{O(|R|)} \cdot n^{O(\Delta n L/\ell)}
    \right),
$
where
the polyhedron is specified intrinsically by a 2D coordinatization
of each face, where each coordinate is a constant-size radical expression
with $b$-bit integers and at most $|R|$ distinct square-roots.

\section{Real RAM Algorithm}
\label{RealRAMAlgorithmSection}

In this section, we give an algorithm to find a closed quasigeodesic on the
surface of a convex polyhedron $P$.
We assume that the input polyhedron is given intrinsically by an isometric
embedding of each face into 2D; Section~\ref{Polyhedral Inputs} describes
reductions from other input representations.
Here we will analyze the algorithm on the real RAM model supporting exact real
arithmetic and square roots (see Section~\ref{sec:real RAM} for details).
Later, in Section~\ref{ExpressionRAMAlgorithmSection}, we will implement
and analyze roughly the same algorithm on the word RAM.

\begin{figure}
\centering
\includegraphics[width=\linewidth]{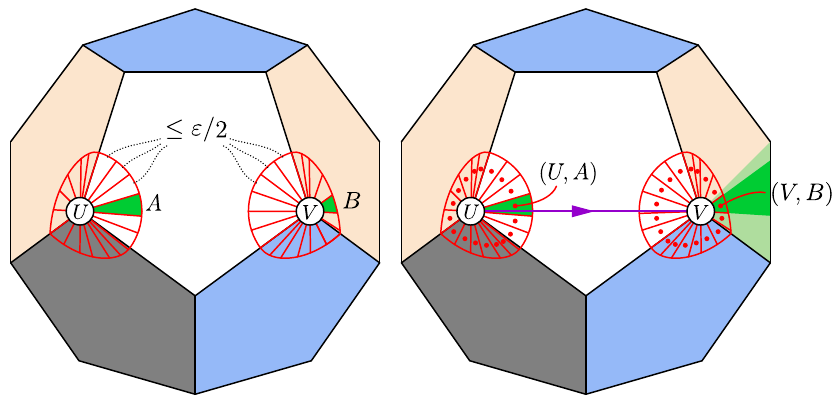}
\caption{We construct a directed graph where vertices are pairs $(U, A)$, where
$U$ is a polyhedron vertex and $A$ is an interval of directions/angles leaving $U$.
The left figure shows two polyhedron vertices, where the space of directions
leaving each vertex is partitioned into intervals of size $\leq \epsilon/2$. The
right figure shows what an edge from $(U, A)$ to $(V, B)$ represents:
there exists a geodesic (purple)
leaving $U$ in a direction from interval $A$ (dark green)
that hits $V$ such that a quasigeodesic may continue from $V$ along any
direction in interval~$B$ (dark green).
The lighter green interval of directions from $V$ containing $B$ shows the full range of angles that the ray may continue while being quasigeodesic.
\looseness=-1}
\label{fig:graph}
\end{figure}

\subsection{Outline}
\label{AlgorithmOutlineSubsection}

First, a bit of terminology: we define a \defn{(quasi)geodesic ray/segment}
to be a one/two-ended path that is (quasi)geodesic.

The idea of the algorithm is roughly as follows.
First, we define a directed graph where each node%
\footnote{We use the word ``node'' and lower-case
letters for vertices of the graph to distinguish them from vertices of a
polyhedron, for which we use capital letters and the word ``vertex''.}
is a pair
$(V, [\vec v_1, \vec v_2])$ of a vertex $V$ of $P$ and a small interval of
directions at $V$, with an edge from one such node $(U,A)$ to another node $(V,B)$
if a geodesic ray starting at the polyhedron vertex $U$ and somewhere in the
interval of directions $A$ can reach $V$ and continue quasigeodesically
everywhere in $B$;\footnote{Because we consider only geodesic rays that can
continue quasigeodesically \emph{everywhere} in $B$, there are some closed
quasigeodesics that we cannot find: those that leave a polyhedron vertex in a
direction in an interval $B$ for which some directions are not quasigeodesic
continuations. In particular, this algorithm is unlikely to find closed
quasigeodesics that turn maximally at a polyhedron vertex.} see
Figure~\ref{fig:graph}.
We show that every node of this graph has out-degree at least~$1$,
and furthermore how to calculate at least one out-edge from
any node of that graph, so we can start at any node and follow directed edges
until hitting a node twice, giving a closed quasigeodesic.
Thus we also obtain a new proof of the existence of a closed quasigeodesic.

The key part of this algorithm is to calculate, given a polyhedron vertex $U$ and a range
of directions as above, another vertex $V$ that can be reached starting from that vertex
and in that range of directions, even though reaching $V$ may require crossing
superpolynomially many faces.


\subsection{Divergence of Geodesic Rays}

In this section, we prove an upper bound on how long we must follow
two geodesic rays before they diverge in the sequence of polyhedron edges
they visit.  First we define some terms.

\begin{definition}
If $X$ is a point on the surface of a polyhedron, $\vec v$ is a direction at
$X$, and $d > 0$, then $S = (X, \vec v, d)$ is the geodesic segment starting at
$X$ in the direction $\vec v$ and continuing for a distance $d$ or until it
hits a polyhedron vertex, whichever comes first.\footnote{This definition is purely geometric;
we reserve calculating these paths for Lemma~\ref{lem:ray follow exact}.} 
We allow $d = \infty$; in that case, $S = (X, \vec v, d)$ may be a geodesic ray
(or it may stop prematurely at a vertex).
\end{definition}

Algorithmically, we will represent a geodesic segment $S = (X, \vec v, d)$
by giving coordinates for $X$ and $\vec v$
in the local coordinate system of some face $f$ containing $X$ ---
usually, the face that the segment $S$ first enters.

\begin{definition}
If $(X, \vec v, d)$ is a geodesic segment or ray, the \defn{edge sequence}
$E(X,\vec v,d)$ is the (possibly infinite) sequence of polyhedron edges that
$(X,\vec v,d)$ visits.
\end{definition}

\begin{lemma}
\label{FaceSequenceDifferenceLemma}
If $S_1 = (X, \vec v_1, \infty)$ and $S_2 = (X, \vec v_2, \infty)$ are two
geodesic rays from a common starting point $X$ with an angle
of $\theta \in (0, \pi)$ between them,
then the edge sequences $E(S_1)$ and $E(S_2)$ are distinct,
and the first difference between them occurs at most one edge after a geodesic
distance of $O(L/\theta)$.
\end{lemma}

\begin{proof}
For a finite distance $d$, the prefix segment $S_i^d = (X, \vec v_i, d)$ is
a straight segment on the unfolded sequence of corresponding faces
intersected by $S_i^d$.
Given a (prefix of) $E(S_i)$, the segment of $S_i$ is a straight line on the unfolded sequence of those faces.
Thus, while $E(S_1^d) = E(S_2^d)$,
the two geodesics $S_1^d$ and $S_2^d$ form a planar wedge in a common unfolding,
as in Figure~\ref{FaceSequenceUnfoldingFigure}.
The distance
between the points on the unfolded rays at distance $d$ from $X$ is $2d \sin
\frac{\theta}{2} > d \theta/\pi$ (since $\frac{\theta}{2} < \frac{\pi}{2}$), so
for points at a distance of $\Omega(L/\theta)$, that distance is at least $L$. So either
$E(S_1)$ and $E(S_2)$ differ before then, or the next edge that $S_1$ and $S_2$
cross is a different edge, in which case $E(S_1)$ and $E(S_2)$ differ in the
next edge, as claimed.
\end{proof}

\begin{figure}
\centering
\includegraphics[width=0.85\linewidth]{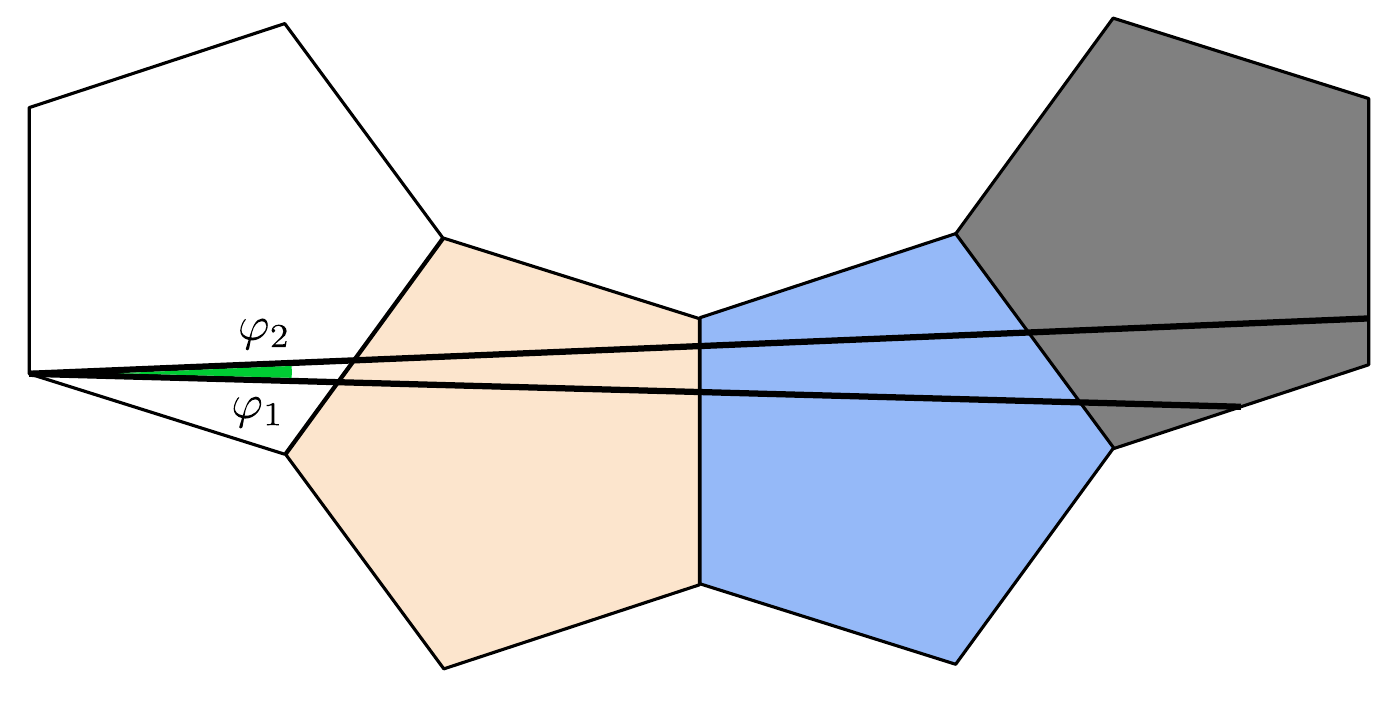}
\caption{A segment of a geodesic is a straight line in the unfolding of the
sequence of faces through which it passes, as in this unfolding of a regular
dodecahedron.}
\label{FaceSequenceUnfoldingFigure}
\end{figure}

If we had defined $L$ analogously to $\ell$ as not just the length of the
longest edge but the greatest distance within a face between a polyhedron vertex and an
edge not containing it, we could remove the ``at most one edge after'' condition
from Lemma~\ref{FaceSequenceDifferenceLemma}.


Lemma~\ref{FaceSequenceDifferenceLemma} gives a bound on the geodesic
distance to the first difference in the edge sequences
(or one edge before that).
We now relate geodesic distance to the number of edges
visited by the geodesic.

\begin{lemma} \label{lem:bounding geodesics}
  Let $S = (X,\vec v, d)$ be a geodesic segment with $d = \Omega(\ell)$.
  Then $E(S)$ consists of $O\big(\frac{d \, L}{\ell^2}\big)$ edges.
\end{lemma}

\begin{proof}
  We prove that, if the geodesic segment $S$ has length $d \leq \ell/4$,
  then $E(S)$ consists of $O(L/\ell)$ edges.
  The lemma then follows by considering $\lceil 4d/\ell \rceil$
  consecutive subsegments along the full geodesic segment.

  Consider the sequence $x_1, x_2, \dots, x_k$ of intersection points
  between the segment $S$ and the respective edges $e_1, e_2, \dots, e_k$
  of the polyhedron.
  Call $x_i$ \defn{near} a vertex $V$ if the intrinsic distance
  $\|x_i - V\| \leq \ell/3$.
  Call the segment $S$ near a vertex $V$ if some $x_i$ is near~$V$.

  We claim that $S$ can be near at most one vertex.
  Assume for contradiction that $S$ has a point $x_i$ near vertex $U$
  and a point $x_j$ near vertex~$V$.
  By the triangle inequality,
  \begin{align*}
  \underbrace{\|U-V\|}_{{}\geq \ell}
  &\leq
  \underbrace{\|U-x_i\|}_{{}\leq \ell/3}
  +
  \underbrace{\|x_i-x_j\|}_{{}\leq d \leq \ell/4}
  +
  \underbrace{\|x_j-V\|}_{{}\leq \ell/3}
  \\
  \ell &\leq (\tfrac{2}{3} + \tfrac{1}{4}) \ell,
  \end{align*}
  a contradiction.

  We can thus divide into two cases, depending on whether $S$ is
  near any (single) vertex:

\begin{figure}
\centering
\includegraphics[width=0.85\linewidth]{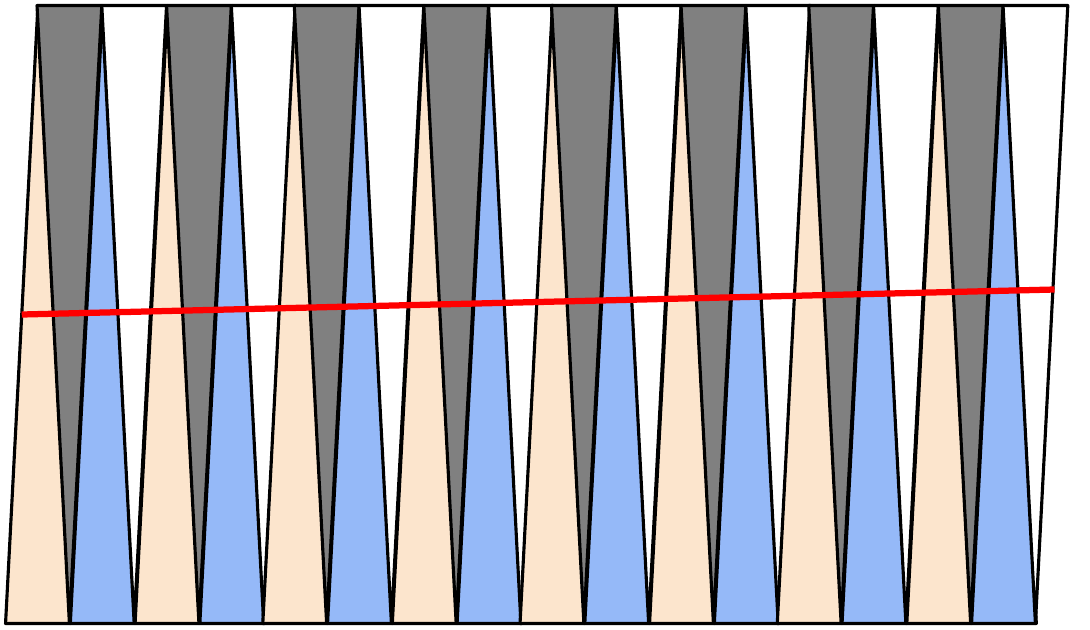}
\caption{If a geodesic path (drawn thick) encounters the same edge twice in nearly the same
place and nearly the same direction, then it may pass the same sequence of faces in the same order a superpolynomial number of times.
In this example, every fourth triangle is the same face, and equally colored faces represent copies of the same face being visited multiple times.}
\label{RepeatedFacesFigure}
\end{figure}

  \paragraph{Case 1: $S$ is far from all vertices.}

  Figure~\ref{RepeatedFacesFigure} shows an example
  where $S$ might cross many faces far from all vertices.
  Consider a segment $x_i x_{i+1}$ crossing a face $f$.
  In this case, $x_i$ and $x_{i+1}$ are not near a vertex.
  Consider shifting this segment, keeping it parallel to $x_i x_{i+1}$,
  until the shifted segment hits a vertex of the face~$f$;
  refer to Figure~\ref{fig:nonvertex bound}.

  \begin{figure}
    \centering
    \includegraphics[width=\linewidth]{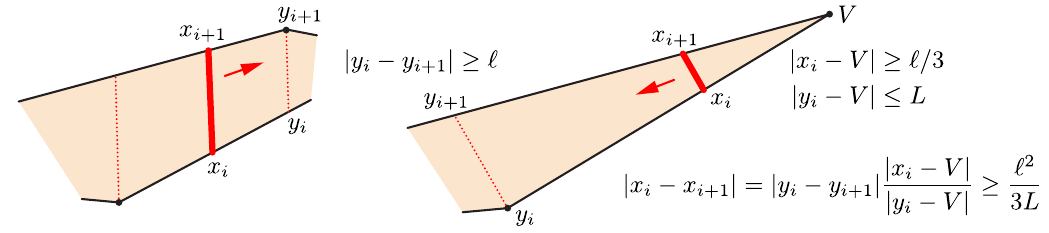}
    \caption{Case 1 of the proof of Lemma~\ref{lem:bounding geodesics}:
      shifting a segment $x_i x_{i+1}$ of the segment $S$ to remain parallel
      while crossing a single face.
      (Left) Case 1.1: shifting in the shrinking direction
      hits a vertex before shrinking to zero length.
      (Right) Case 1.2: shifting in the growing direction
      when the shrinking direction degenerates to a single vertex~$V$.}
    \label{fig:nonvertex bound}
  \end{figure}

  \paragraph{Case 1.1:}
  If one shift direction causes the segment to get shorter
  and the shifted segment $y_i y_{i+1}$ hits a vertex
  before it becomes zero length
  (as in Figure~\ref{fig:nonvertex bound}, left),
  then $y_i y_{i+1}$ has length at least $\ell$ (by definition of~$\ell$),
  and thus $x_i x_{i+1}$ has length at least $\ell \geq \ell^2/(3 L)$.

  \paragraph{Case 1.2:}
  Otherwise, we know that there is a shift direction where the edge
  first hits a vertex $V$ of $f$ when it collapses to zero length at~$V$,
  which is a common endpoint of edges $e_i$ and $e_{i+1}$ of $f$
  that $x_i$ and $x_{i+1}$ lie on
  (as in Figure~\ref{fig:nonvertex bound}, right).
  Now shift the edge in the opposite direction, which lengthens the edge,
  until we obtain an edge $y_i y_{i+1}$ where either $y_i$ or $y_{i+1}$
  is a vertex of~$f$.
  The shifted segment $y_i y_{i+1}$ has length at least $\ell$
  (by definition of~$\ell$).
  Segments $x_i x_{i+1}$ and $y_i y_{i+1}$ form similar triangles with~$V$.
  We have either $\|y_i - V\|$ or $\|y_{i+1} - V\| \leq L$
  (whichever is an edge, by definition of~$L$)
  and both $\|x_i - V\| \geq \ell/3$ and $\|x_{i+1} - V\| \geq \ell/3$
  (by farness), so the coefficient of similarity is at most $L/(\ell/3)$.
  Therefore $x_i x_{i+1}$ has length at least
  $\ell/(L/(\ell/3)) = \ell^2/(3 L)$.

  \medskip

  Hence, in both Case 1.1 and 1.2,
  $x_i x_{i+1}$ has length at least $\ell^2/(3 L)$.
  But $S$ has length $\leq \ell/4$.
  Therefore, in Case 1, the number of polyhedron edges hit by segment $S$
  is at most $(\ell/4)/(\ell^2/(3 L)) = O(L/\ell)$.

  \paragraph{Case 2: $S$ is near a vertex $V$.}
  (This case includes when $S$ starts or ends at a vertex~$V$.)
  First consider a segment $x_i x_{i+1}$ where one endpoint (say, $x_i$)
  is near $V$ while the other endpoint (say, $x_{i+1}$) is far from~$V$.
  By the triangle inequality,
  \begin{align*}
    \|V - x_{i+1}\|
    &\leq \underbrace{\|V - x_i\|}_{{}\leq \ell/3}
        + \underbrace{\|x_i - x_{i+1}\|}_{{}\leq \ell/4}
    \\
    &< \ell.
  \end{align*}

  Thus any segment $x_i x_{i+1}$ in Case~2 must have both $x_i$ and $x_{i+1}$
  within distance $< \ell$ of vertex~$V$.
  Hence, by definition of $\ell$,
  the edges $e_i$ hit by $S$ are all incident to~$V$.
  Consider the sequence of faces $f_0, f_1, \dots, f_k$ intersected by $S$,
  where $f_i$ is the face between intersections $x_i$ and $x_{i+1}$
  for $1 \leq i < k$;
  $f_0$ is the face before the intersection $x_0$ (if any); and
  $f_k$ is the face after the intersection $x_k$ (if any).
  As argued above, the $f_i$ all have $V$ as a vertex.
  Thus, if we unfold these faces consecutively into the plane,
  as shown in Figure~\ref{fig:ManyFaces3},
  then the unfolded faces all rotate around a common point~$V$.
  In this unfolded view, $S$ is a straight segment
  passing through the collinear points $x_1, x_2, \dots, x_k$.
  Each angle $\angle x_i, V, x_{i+1}$ is an angle of a face (at~$V$),
  which is at least $\arcsin(\ell/L) \geq \ell/L$;
  see Figure~\ref{fig:sin bound}.
  Because the unfolding is planar,
  $\sum_{i=1}^{k-1} \angle x_i, V, x_{i+1}
  = \angle x_1, V, x_k \leq 180^\circ$.
  Thus $k \leq 180^\circ / (\ell/L) = O(L/\ell)$,
  so the number of faces hit by $S$ is $O(L/\ell)$.
\end{proof}

\begin{figure}
  \centering
  \begin{subcaptionblock}{\linewidth}
    \centering
    \includegraphics[width=\linewidth]{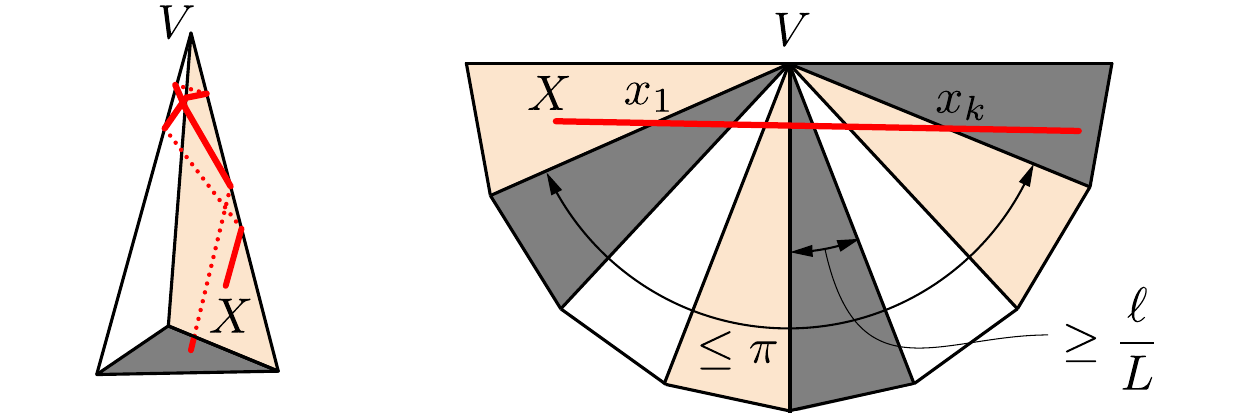}
    \caption{The geodesic is straight so the total angle of visited faces must be at most~$\pi$.}
    \label{fig:ManyFaces3}
  \end{subcaptionblock}
  \begin{subcaptionblock}{\linewidth}
    \centering
    \includegraphics{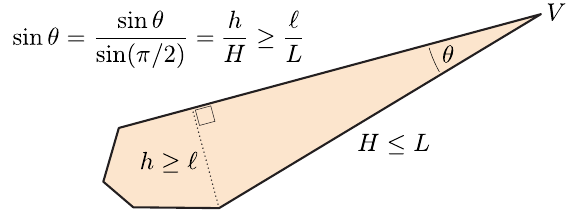}
    \caption{Each angle $\theta$ of a face is at least $\ell/L$.}
    \label{fig:sin bound}
  \end{subcaptionblock}
  \caption{A short quasigeodesic can visit $\omega(1)$ faces
    but $O(L/\ell)$ faces near a vertex $V$.}
\end{figure}

\subsection{Computing Quasigeodesic Rays}

Next we show how to algorithmically follow quasigeodesic rays.
First we need a lemma about locally unfolding the polyhedron's surface.

\begin{lemma} \label{lem:coordinate system transforms}
  Given the intrinsic coordinates $C_1,C_2$ for two adjacent faces $f_1,f_2$
  sharing an edge~$e$, where each $C_i$ is a vector specifying
  the planar coordinates of all vertices of~$f_i$,
  we can compute the orientation-preserving isometry
  bringing $e$ on $f_1$ to $e$ on $f_2$ as a transformation matrix~$I$.
  On the real RAM, the running time is $O(1)$.
\end{lemma}

\begin{proof}
  Let $e_i = (u_i, v_i)$ be the embedding of edge $e$ in the intrinsic
  coordinates $C_i$ of $f_i$, where $u_1$ and $u_2$ correspond to the
  same vertex.
  The orientation-preserving isometry $T$ bringing $e_1$ to $e_2$
  can be written as the composition of a translation by $-u_1$,
  followed by a rotation $R$, followed by a translation by $+u_2$.
  The rotation $R = \left(\begin{smallmatrix}
    c & s \\
    -s & c \\
  \end{smallmatrix}\right)
  $ is given by the 2-equation 2-unknown linear system
  $$
  \begin{pmatrix}
    c & s \\
    -s & c \\
  \end{pmatrix}
  (v_1 - u_1)
  =
  v_2 - u_2.
  $$
  This system has a solution, namely $c = \cos \theta$ and $s = \sin \theta$
  where $\theta$ is the rotation angle between the vectors.
  We can solve this system using $O(1)$ arithmetic operations on
  the coordinates of $e_1$ and~$e_2$, e.g., using Gaussian elimination.%
  \footnote{Surprisingly, this direct transformation is cleaner computationally
    than transforming each edge to/from a canonical location on the $x$ axis,
    which would require computing the length of the vectors,
    which involves a square root.}
  Finally, we obtain the desired affine isometry by multiplying the
  $3 \times 3$ translation, rotation, and translation matrices
  operating on homogeneous coordinates.
\end{proof}

\begin{lemma} \label{lem:ray follow exact}
  Let $S = (X, \vec v, \infty)$ be a geodesic ray
  on an $n$-vertex polyhedron,
  where point $X$ is on the boundary of a face $f$,
  $\vec v$ points inside $f$ from $X$, and
  $X$ and $\vec v$ are specified in the local coordinate system of~$f$.
  We can compute the first $k$ faces and $k$ edges visited by~$S$;
  the corresponding $k$ intersection points; and
  a planar embedding of an unfolding of these faces and edges.
  In particular, this determines the direction by which $S$ exits the
  last face at the last intersection point.
  If $S$ hits a vertex within the first $k$ steps, we stop there and
  also output the vertex; otherwise, we proceed for exactly $k$ steps.
  On a real RAM, the running time is $O(k \lg n)$.
\end{lemma}

\begin{proof}
  Suppose that $S$ enters face $f_i$
  at point $Q_{i-1} = (Q_x,Q_y)$ with direction $\vec v_{i-1} = (v_x,v_y)$,
  where $Q_i$ and $\vec v_{i-1}$ are in the local coordinate system of~$f_i$.
  (To start, $Q_0 = X$, $\vec v_0 = \vec v$, and $f_1 = f$.)
  We binary-search to find an edge of $f_i$ where $S$ exits~$f_i$.
  For each visited edge $e=(A,B)$ of $f_i$
  with endpoints $(A_x,A_y)$ and $(B_x,B_y)$,
  the intersection $Q'$ of the line extension of $S$,
  which has equation $$(Q'_x-Q_x) \cdot v_y = v_x \cdot (Q'_y-Q_y),$$
  and the line extension of~$e$,
  which has equation $$(Q'_x-B_x)\cdot (A_y-B_y) = (A_x-B_x)\cdot (Q'_y-B_y),$$
  is the point $Q' = (Q'_x,Q'_y)$ given by the solution to the
  above linear system.
  To determine whether $S$ in fact crosses~$e$,
  we compute the signed areas of triangles $\triangle Q_{i-1} Q' A$
  and $\triangle Q_{i-1} Q' B$
  and compare them to zero to determine their signs.
  If the two areas have opposite sign, or either sign is zero, then
  we have found a desired edge~$e_i = e$
  and the resulting exit point~$Q_i = Q'$.
  Otherwise, the shared sign of the triangles tells us
  whether the intersection $Q'$ is before $A$ or after $B$ on $e$,
  and we direct the binary search accordingly.
  This search takes $O(\lg n)$ time.

  If $S$ hits a vertex of face $f_i$, we will determine so by detecting
  that it hits two edges of~$f_i$.
  Otherwise, for the correct edge $e_i$,
  we can write $(Q_i, \vec v_{i-1})$
  as an affine transformation $\rho_i$ applied to $(Q_{i-1}, \vec v_{i-1})$,
  namely, $\rho_i$ is the translation by vector $Q_i - Q_{i-1}$.
  We then take Lemma~\ref{lem:coordinate system transforms}'s
  affine transformation $\sigma_i$ that brings edge $e_i$ of $f_i$
  in the local coordinate system of $f_i$ to the corresponding edge $e_i$
  of the adjacent face $f_{i+1}$
  in the local coordinate system of $f_{i+1}$,
  and apply it to both $Q_i$ and $\vec v_{i-1}$
  to obtain the point $Q_i$ at which and the direction $\vec v_i$ by which
  that ray $S$ will next enter $f_{i+1}$ in its local coordinate system.
  Composing $\rho_i$ and $\sigma_i$, we obtain $(Q_i, \vec v_i)$
  as an affine transformation $\tau_i$ of $(Q_{i-1}, \vec v_{i-1})$,
  in $O(1)$ time.


  Applying this method repeatedly, we can obtain the point $Q_i$ and
  direction $\vec v_i$ after $S$ traverses $i$ faces in sequence,
  where $(Q_i, \vec v_i)$ is given by $\tau_i$
  applied to $(Q_{i-1}, \vec v_{i-1})$.
  We compute each $(Q_i, \vec v_i)$ vector sequentially from the
  previous one, in $O(\lg n)$ time because of the binary search.
  The total running time for $k$ such steps is thus $O(k \lg n)$.

  Finally, we compute a planar embedding of an unfolding of the faces
  $f_1, f_2, \dots, f_k$ and edges $e_1, e_2, \dots, e_k$
  visited by the ray~$S$.
  We transform each
  $f_i, e_i$ from the local coordinate system of $f_i$
  to a common local coordinate system of $f_1$,
  by applying the inverse affine transformations
  $\sigma^{-1}_{i-1}, \sigma^{-1}_{i-2}, \dots, \sigma^{-1}_1$
  in that order.
\end{proof}

\begin{corollary} \label{cor:cone following}
  Consider an angle-$\theta$ cone between geodesic rays
  $S_1 = (X,\vec v_1, \infty)$ and $S_2 = (X,\vec v_2, \infty)$.
  We can compute a geodesic segment $S = (X, \vec v, d)$
  that is in the given cone and ends at a vertex $Y$ of~$P$
  and has $k = O\left(\frac{L^2}{\theta \, \ell^2}\right)$
  intersections between $S$ and edges of~$P$,
  along with the identity of the faces and edges intersected by~$S$.
  The output consists of $\vec v$, $d$, and the $k$ intersections.
  On a real RAM, the running time is $O(k \lg n)$.
\end{corollary}

\begin{proof}
  %
  We apply Lemma~\ref{lem:ray follow exact}
  to follow the given geodesic rays $S_1$ and $S_2$
  until they cross different edges.
  As a special case, if $S_1$ and $S_2$ immediately enter different faces,
  then $X$ is on an edge of $P$ and the cone contains one of the endpoints
  of that edge, and $S$ can simply be a portion of that edge.
  Otherwise, $S_1$ and $S_2$ initially enter the same face, and each step,
  either they exit the face along the same edge (and thus enter the same face)
  or they exit along different edges.
  By Lemma~\ref{FaceSequenceDifferenceLemma},
  after a distance of $d = O(L/\theta)$,
  $S_1$ and $S_2$ must exit a common face $f$ along different edges,
  say $e_1$ and~$e_2$.
  By Lemma~\ref{lem:bounding geodesics},
  this event happens after crossing $O\left(\frac{d \, L}{\ell^2}\right) =
  O\left(\frac{L^2}{\theta \, \ell^2}\right)$ edges of~$P$.
  Thus we can find the desired vertex $Y$ by choosing a vertex
  between $e_1$ and $e_2$ on~$f$
  (where ``between'' is defined by the sides of the rays
  corresponding to the cone).
  Lemma~\ref{lem:ray follow exact} also provides a planar embedding of an
  unfolding of the sequence of faces and edges visited by $S_1$ and $S_2$
  up to and including~$f$.
  Thus we can draw the unfolded segment $S$ from $X$ to $Y$ in this embedding,
  and intersect with each of the unfolded edges to find the
  intersection points along the way (which can be mapped back to the
  polyhedron, if desired, via the transformations provided by
  Lemma~\ref{lem:ray follow exact}).
\end{proof}

\subsection{Full Algorithm}

We are now ready to state the algorithm for finding a closed quasigeodesic in
pseudopolynomial time:

\begin{theorem}
\label{FindQuasigeodesicTheorem}
Let $P$ be a convex polyhedron with $n$ vertices all of curvature at least
$\epsilon$, let $L$ be the length of the longest edge, and let $\ell$ be the
smallest distance within a face between a vertex and a nonincident edge.
Then we can find a closed quasigeodesic on $P$ consisting of
$O\left(\frac{n \, L^2}{\epsilon^2 \, \ell^2}\right)$
segments on faces of~$P$.
On a real RAM, the running time is
$O\left(\frac{L^2}{\epsilon^2 \, \ell^2} \, n \lg n\right)$.
\end{theorem}

\begin{proof}
First, we represent the minimum curvature $\epsilon$ of the polyhedron's
vertices by computing a unit vector $\vec v_\epsilon$
whose counterclockwise angle from the positive $x$ axis is~$\epsilon$.
For each vertex~$V$, we compute a planar embedding of the faces
$f_1, f_2, \dots, f_k$ incident to $V$ in clockwise order, as follows.
Let $e$ be the edge shared by $f_1$ and $f_k$.
We apply Lemma~\ref{lem:coordinate system transforms} to find a transformation
from $f_1$'s local coordinate system that maps $V$ to the origin and
maps $e$ to the positive $x$~axis,
with $f_1$ below the axis,
by constructing an artificial local coordinate system
that places $V$ and $e$ in this way.
Then we repeatedly apply Lemma~\ref{lem:coordinate system transforms}
to place the subsequent faces $f_2, \dots, f_k$, aligning corresponding edges
(similar to Lemma~\ref{lem:ray follow exact}).
The computed placement of $f_k$ gives an embedding of $e$
which, viewed as a vector from $V$, forms a counterclockwise angle
with the $x$~axis equal to the curvature of vertex~$V$.
Dividing this vector by $\|e\|$ gives us a unit vector.
On the real RAM, the vector computation at each vertex $V$ costs $O(\deg(V))$,
for a total of $O(n)$ time by the Handshaking Lemma and planarity.
We can compare two unit vectors computed for two different vertices
by comparing the slopes of the vectors (the ratio of the two coordinates)
and the signs of the two coordinates (determining the quadrant).
Via a linear scan, we can
find the vector $\vec v_\epsilon$ with the smallest angle in the best quadrant,
whose counterclockwise angle from the positive $x$~axis
is the minimum angle~$\epsilon$.
On the real RAM, this scan costs $O(n)$ time.

Second, we bisect the angle $\epsilon$,
by computing a unit vector $\vec v_{\epsilon/2}$
whose counterclockwise angle from the positive $x$ axis is~$\epsilon/2$.
There are three cases.
When $\vec v_\epsilon$ has strictly positive $y$ coordinate,
we take $\vec v_{\epsilon/2}$ to be the average of vectors
$\vec v_\epsilon$ and $(1,0)$,
and normalize this vector to obtain a unit vector.
When $\vec v_\epsilon$ has strictly negative $y$ coordinate,
we take $\vec v_{\epsilon/2}$ to be the negation of this average.
When $\vec v_\epsilon$ has zero $y$ coordinate,
we take $\vec v_{\epsilon/2}$ to be $(0,1)$.
Then we repeat this process to construct a unit vector $\vec v_{\epsilon/4}$
whose counterclockwise angle from the positive $x$ axis is~$\epsilon/4$.
On the real RAM, this step takes constant time.


Third, we construct $k = \left\lfloor \frac{2\pi}{\epsilon/4} \right\rfloor$
unit direction vectors $\vec d_0, \vec d_1, \dots, \vec d_{k-1}$
such that the counterclockwise angle between consecutive vectors
$\vec d_i, \vec d_{i+1}$ is exactly $\epsilon/4$,
and the wraparound pair $\vec d_{k-1}, \vec d_0$
has angle between $\epsilon/4$ and $\epsilon/2$.
We start with $\vec d_0 = (1,0)$ (the positive $x$ axis)
and $\vec d_1 = \vec v_{\epsilon/4}$.
Then, for $i = 2, 3, \dots$, we define the $i$th direction vector $\vec d_i$
to be the rotation of $\vec d_{i-1}$ counterclockwise by angle $\epsilon/4$.
Thus $\vec d_i = (x_i, y_i)$ forms angle $i \epsilon/4$
with the positive $x$ axis, so it can be computed using
cosine/sine sum formulas:
\begin{align*}
  \vec d_i = (x_i, y_i)
  &=\left(\cos i \tfrac{\epsilon}{4}, ~ \sin i \tfrac{\epsilon}{4}\right) \\
  &=\left(
      \cos \left[ (i-1) \tfrac{\epsilon}{4} + \tfrac{\epsilon}{4} \right], ~
      \sin \left[ (i-1) \tfrac{\epsilon}{4} + \tfrac{\epsilon}{4} \right]
    \right) \\
  &=\left(
      \cos (i-1) \tfrac{\epsilon}{4} \cos \tfrac{\epsilon}{4} -
      \sin (i-1) \tfrac{\epsilon}{4} \sin \tfrac{\epsilon}{4}, ~
      \sin (i-1) \tfrac{\epsilon}{4} \cos \tfrac{\epsilon}{4} +
      \cos (i-1) \tfrac{\epsilon}{4} \sin \tfrac{\epsilon}{4}
    \right) \\
  &=\left(
      x_{i-1} x_1 - y_{i-1} y_1, ~
      y_{i-1} x_1 + x_{i-1} y_1
    \right).
\end{align*}
At each step, we check whether we have completed a full cycle around the origin
by testing whether $y_{i-1}$ is negative and $y_i$ is positive.
In this case, $i-1 = k = \left\lfloor \frac{2\pi}{\epsilon/4} \right\rfloor$
and we are done (discarding the last two vectors).
On the real RAM, this computation costs $O(1/\epsilon)$ time.

Fourth, for each vertex $V$ of each face $F$ of~$P$,
we choose $O(1/\epsilon)$ direction vectors in $F$'s local coordinate system
that include the two incident edges of~$F$,
and so that the angle of the wedge between consecutive direction vectors
is at most~$\epsilon/2$.
Namely, we take the two incident edge direction vectors $\vec e_1, \vec e_2$
and the subset of the $\vec d_i$ vectors that point within the face $F$ at~$V$.
We can detect whether a vector $\vec d_i$ points within $F$
by comparing the slope of $\vec d_i$ with the slope of each $\vec e_j$
and evaluating the signs of each vector component to determine the quadrant.
On the real RAM, this computation costs
$O(\deg(V)/\epsilon)$ time
for each vertex~$V$, for a total of $O(n/\epsilon)$ by the Handshaking Lemma
and planarity.

As preparation for the fifth step,
define a directed graph $G$ with one node for
each pair $(V,A)$ of a vertex $V$ from $P$
and one of its wedges $A$ (for ``angular range''),
giving the graph $O(n/\epsilon)$ nodes
(by the previous analysis).
Graph $G$ has an edge from a node
$s = (U, A)$ to a node $t = (V, B)$
if there exists a direction in $A$ such that
a quasigeodesic ray starting in that direction from the polyhedron vertex $U$
and hitting the polyhedron vertex $B$
can continue from every direction in~$V$. 

Fifth, we describe how to find an outgoing edge from any given node
$s = (U, A)$ in~$G$, with corresponding vertex $U$ and
a wedge $A$ of directions from $\vec v_1$ to $\vec v_2$. 
Refer to Figure~\ref{fig:graph}.
If wedge $A$ contains a polyhedron edge $(U,V)$
(by our construction, this edge will be in the direction $\vec v_1$
or $\vec v_2$ from~$U$),
then we will proceed to polyhedron vertex $V$ using that direction vector
$\vec v_f$, which we view in the local coordinate system of one of the
faces $f$ incident to the polyhedron edge $(U,V)$.
Otherwise, apply Corollary~\ref{cor:cone following} to follow the cone between
geodesic rays $S_1 = (U, \vec v_1, \infty)$ and $S_2 = (U, \vec v_2, \infty)$
to find a vertex $V$ within the cone, i.e., reachable by a quasigeodesic
segment $S$ starting at $U$ within the wedge~$A$,
along with the final face $f$ visited by $S$
and the final direction vector $\vec v_f$ of the segment $S$
in the local coordinate system of~$f$.
In either case, we obtain a quasigeodesic that can then exit the vertex $V$
anywhere in a cone with angle equal to that vertex's curvature,
which is at least $\epsilon$, so for
at least one of the wedges $B$ of size $\leq \epsilon/2$ at~$V$,
the quasigeodesic can exit anywhere in that wedge.
To find such a wedge, we construct a planar embedding of the
polyhedron faces incident to $V$ (as in the first paragraph of the proof)
in counterclockwise order from~$f$, and then extend the vector $\vec v_f$
from the embedded $V$ and compute (via binary search) which face $f'$ we enter.
If there is no such face $f'$, then we instead construct a planar embedding
of the polyhedron faces incident to $V$ in \emph{clockwise} order from~$f$,
and again extend the vector $\vec v_f$ and compute (via binary search)
which face $f'$ we enter.
If there is still no such face $f'$, then we can safely enter any face
other than $f$ in any direction at~$V$
(as these other faces can be rotated around $V$ so that a chosen direction
becomes the straight extension of $\vec v_f$),
so we can choose any corresponding node $(V, B)$ of the graph.
Otherwise, assuming $f'$ exists in one of the two layouts,
transform the vector $\vec v_f$ to the local coordinate system
of~$f'$, and find the clockwise next whole wedge $B$ in $f'$;
if there is no such whole wedge within $f'$,
choose the clockwise first wedge $B$ in the clockwise next face after~$f'$.
Rounding to the clockwise next wedge bends by at most $\epsilon/2$,
and leaving at a ray within the wedge bends by at most another $\epsilon/2$,
for a total bend of at most $\epsilon$, so all rays within the exit wedge $B$
form valid quasigeodesic extensions.
Thus we find an outgoing edge from node $s = (U, A)$ to node $t = (V, B)$.
The running time on the real RAM is
$O(k \lg n + \lg n) = O(k \lg n)$
where $k = O\left( \frac{L^2}{\theta \, \ell^2} \right)$
and by construction $\theta \geq \epsilon/4$.
Therefore the running time is
$O\left(\frac{L^2}{\epsilon \, \ell^2} \lg n\right)$.

Sixth, we start from any node of $G$, and repeatedly traverse outgoing edges
of $G$ until we repeat a node of~$G$.
In the worst case, we compute an outgoing edge for each of the
$O(n/\epsilon)$ nodes of $G$ before finding a cycle.
The resulting cycle in $G$ corresponds to a closed quasigeodesic
on the polyhedron, by the definition of the graph.
Each traversal of an edge in $G$ corresponds to a vertex-to-vertex geodesic
path, which either follows a polyhedron edge so has just one segment,
or applies Corollary~\ref{cor:cone following} to a cone
with angle $\geq \epsilon/4$
(by the third paragraph, because the cone did not contain a polyhedron edge),
so has $O\left(\frac{L^2}{\epsilon \, \ell^2}\right)$ segments.
The closed geodesic can be described by $O(n/\varepsilon)$ such
vertex-to-vertex paths, for a total of $O\left(\frac{n \, L^2}{\epsilon^2 \, \ell^2}\right)$
segments and
$O\left(\frac{L^2}{\epsilon^2 \, \ell^2} \, n \lg n\right)$ time.

On the real RAM,
the overall running time is $O(n)$ for the preprocessing in the
first paragraph,
$O(1)$ for the preprocessing in the second paragraph,
$O(1/\epsilon)$ for the preprocessing in the third paragraph,
$O(n/\epsilon)$ for the preprocessing in the fourth paragraph, and
$O\left(\frac{L^2}{\epsilon^2 \, \ell^2} \, n \lg n\right)$ for
walking/constructing the graph, of which the last term dominates.
\end{proof}

If $D$ is the greatest diameter of a face, then a closed quasigeodesic found by
Theorem~\ref{FindQuasigeodesicTheorem} has length $O\left({n \over \epsilon}
\left({L \over \epsilon} + D\right)\right)$, because the quasigeodesic visits $O(n/\epsilon)$
graph nodes and, by Lemma~\ref{FaceSequenceDifferenceLemma}, goes a distance at
most $L/\epsilon + D$ between each consecutive pair.  

If we define the graph $G$ in terms of $\epsilon/s$ instead of $\epsilon/2$
for $s \geq 2$,
then we can guarantee that, at each vertex, the closed quasigeodesic has
at most $180^\circ - \epsilon/2 + \epsilon/s$ material on either side.
In particular, for $s > 2$,
we guarantee strictly less than $180^\circ$ of material
on either side of the closed quasigeodesic at each vertex.

\section{Models of Computation}
\label{sec:model}

In this section, we review standard models of computation and
introduce a new model --- the expression RAM --- that makes it
easy to state an algorithm that manipulates real numbers
while grounding it in computational feasibility.
The expression RAM is essentially a form of the real RAM,
and indeed the expression RAM can be implemented on top of the real RAM
(except for possible lack of floor/ceiling operations).
More interesting is that the expression RAM can be implemented on top of
the word RAM, with different operation costs; see Figure~\ref{fig:models}.
Thus we can express a single algorithm in the expression RAM, and then
analyze its running time on both the real RAM and word RAM.

\begin{figure}[htbp]
  \centering
  \includegraphics[scale=0.35]{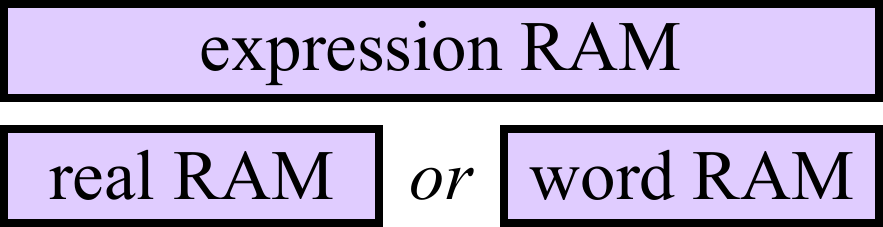}
  \caption{The expression RAM model can be implemented on top of the
    real RAM or the word RAM (with different operation costs).}
  \label{fig:models}
\end{figure}

Our approach is essentially the ``exact geometric computation'' framework
of LEDA/CGAL reals \cite{SepBound} and CORE reals \cite{CORE}.
So while the model and associated algorithms have essentially been described
before (and even implemented in code), they have not previously been analyzed
to the level of precise costs in the word-RAM model.
We expect that the expression RAM can be applied to analyze many algorithms
in computational geometry, both old and new, making it a contribution
of independent interest.

\subsection{Standard Models: Real RAM and Word RAM}
\label{sec:real RAM}

First we recall the two most common models of computation in
computational geometry and data structures, respectively:
the real RAM and the word RAM.

Both models define a \defn{Random Access Machine (RAM)} to have a 
\defn{memory} consisting of an array $M$ with cells $M[0], M[1], \dots$
(where the first $O(1)$ cells act as ``registers''),
but they differ in what the cells can store.
The RAM supports the $O(1)$-time operations
``$M[i] = M[M[j]]$'' (read) and ``$M[M[j]] = M[i]$'' (write)
where $i$ and $j$ are constant integers (e.g., $7$ and~$42$)
representing register indices,
and $M[j]$ is promised to be an integer
(or else the operation fails).
When we later allow operations on cell values, such as ``$a = b + c$'',
these operations are actually restricted on a RAM to be of the form
``$M[i_1] = M[i_2] + M[i_3]$''
where $i_1,i_2,i_3$ are all constant integers
(again representing register indices).

\label{sec:real RAM floor}
The \defn{real RAM} model
(dating back to Shamos's original thesis \cite{Shamos-1978})
allows storage and manipulation of black-box real numbers
supporting arithmetic operations $+,-,\times,\div$,
radical operations $\sqrt[d]{\phantom{\imath}}$, and comparisons $<, >, =$
in constant time,
in addition to storage and manipulation of integers used for indexing.%
\footnote{Shamos's original description of the real RAM
  \cite{Shamos-1978} supports the more general operations
  $\exp$ and $\ln$, from which one can derive $\sqrt[d]{\phantom{\imath}}$,
  but these features do not seem to be in common use in computational geometry.
  The original description also supports trigonometric functions,
  but this has fallen out of favor because of related undecidability results
  \cite{Laczkovich-2003} described below.
}
%
While popular in computational geometry for its simplicity, this model is not
directly realistic for a digital computer, because it does not bound the
required precision, which can grow without bound.
For example, the real RAM model crucially does not support converting
black-box real numbers into integers (e.g., via the floor function),
or else one can solve PSPACE \cite{s-pram-79} and
\#SAT \cite{bms-scram-85} in polynomial time.
The real RAM (without floor) is one of the models we will use to analyze
our algorithm.
(Refer ahead to Table~\ref{tab:ops}, last column, for the exact operations
we allow on a real RAM.)

The \defn{word RAM} \cite{Fredman-Willard-1993}
allows storage and manipulation of
$w$-bit integers (called \defn{words}) supporting arithmetic
($+, \allowbreak -, \allowbreak \times, \allowbreak \div,
\allowbreak \operatorname{mod}$), bitwise operations
($\textsc{and}, \allowbreak \textsc{or}, \allowbreak \textsc{xor},
\allowbreak \textsc{not}$), and comparisons ($<, \allowbreak >, \allowbreak =$)
in constant time.
Furthermore, the model makes the \defn{transdichotomous assumption}
that $w \geq \lg n$ where $n$ is the problem size
(named for how it bridges the problem and machine \cite{Fredman-Willard-1993}).
This assumption ensures that a single word can address the memory required
to store the $n$ inputs; it is essentially the minimal assumption that
suffices to enable efficient random access on a RAM.
Technically, we assume that the real RAM also includes all word RAM operations
(i.e., the real RAM can store and manipulate both black-box real numbers
and $w$-bit words where $w \geq \lg n$),
so that it can support random access and standard data structures.%
\footnote{Shamos's original description of the real RAM
  \cite{Shamos-1978} says that it supports
  ``Indirect accessing of memory (integer addresses only)'',
  presumably with the intent that such addresses should be formed via
  arithmetic on integers (as opposed to real arithmetic that happens to
  result in integers, which would be difficult to detect).
  Adding the entire uncontroversial word RAM model (which did not exist in 1978)
  enables more flexible address manipulation (e.g., word division and mod),
  so that the real RAM is just about adding black-box real numbers.}


The word RAM is the standard model of computation in data structures and much
of theoretical computer science, capturing the word-level parallelism in
real computers.  But it is highly restrictive for computational geometry
because it restricts the \emph{inputs} to be integers or rationals, whereas we
normally want to allow continuous real inputs (e.g., a point set or a polygon).
Recent work starting with \cite{Chan-Patrascu-I,Chan-Patrascu-II} achieves
strong results in computational geometry on the word RAM, but not all problems
adapt well to integer or rational inputs.
For example, in our quasigeodesics problem, should the input consist of
vertex coordinates restricted to be integers, or an intrinsic polyhedron
metric restricted to have integer edge lengths?
Each answer defines a different subset of problems,
and neither obviously captures all instances that we might want to solve.
For example, the regular dodecahedron and icosahedron cannot be represented
with rational vertex coordinates (and in 4D, rational vertex coordinates forbid
getting even a desired graph structure) \cite{Ziegler-2008}.

Our solution is to build a new model called the ``expression RAM''.
On the one hand, this model has a built-in notion of restricted real numbers,
thereby encompassing a form of the real RAM.
This enables problem inputs to be given as certain real numbers,
including integers and rationals but also roots
(though these incur a cost in any computation involving them).
As we show in Section~\ref{Polyhedral Inputs}, this model provides reductions
between the different possible representations of polyhedral inputs,
though at some cost.
On the other hand, our model can be implemented directly on
the word RAM (or the real RAM) with appropriate operation costs,
so its algorithms can be realistically run on digital computers.
We believe our model is the first to unify the simplicity of the real RAM with
guaranteed algorithmic performance as measured on the word RAM,
though our approach is closely based on the theory and practice of
exact geometric computation \cite{SepBound,CORE}.

%

\subsection{Expression RAM}
\label{sec:Expression RAM}


The \defn{expression RAM (Random Access Machine)} is a RAM model of computation
allowing storage and manipulation of real radical expressions over integers,
called ``expressions'' for short.
More precisely, an \defn{expression} represents a real number as
an ordered DAG with a specified source node%
\footnote{We use the term ``source'' instead of the more standard ``root''
  to avoid confusion with nodes that represent roots in radical expressions.}
(representing the overall value of the expression), where
each leaf corresponds to an integer,
each two-child node corresponds to a binary operator among $\{+,-,\times,\div\}$,
each one-child node corresponds to a unary operator $\sqrt[d]{\phantom{\imath}}$
for an integer $d$ stored within the node.%
%
\footnote{In our quasigeodesics algorithm, we only need square roots
  ($d=2$), but we define general fixed roots
  for potential future uses in other algorithms.
  %
  Another potential addition to the model is the expression
  ``$j$th smallest real root of the degree-$d$ polynomial with coefficients
  given by $d+1$ expressions'' \cite{SepBound,CORE}, but this would require
  additional work which we do not carry out here.}
%
A special case of an expression is an \defn{integer expression} which consists
of a single integer (leaf) and no operators.
Like the real RAM, we define the expression RAM to extend the word RAM,
so it also allows storage and manipulation of $w$-bit words
where $w$ is a model parameter satisfying $w \geq \lg n$.
To unify the two data types, we define a $b$-bit integer expression to be
explicitly encoded as $\lceil b/w \rceil$ words of $w$ bits,
which can be manipulated as usual by word RAM operations.
In particular, we can convert words into $w$-bit integer expressions
and vice versa in $O(1)$ time.

Table~\ref{tab:ops} defines the expression-RAM operations for forming
and evaluating expressions.
Constructing expressions (Operation~\ref{op:combine}) is essentially free,
but also does not compute anything other than an expression DAG,
which can then be input into a computation
(Operations~\ref{op:real}, \ref{op:B}, \ref{op:realbits}, and~\ref{op:floor}).
Roughly speaking, Operation~\ref{op:realbits} computes the integer part and
the first $b$ bits of the fractional part of $E$ (useful e.g.\ for
reporting or plotting approximate values), but the last bit may be incorrect.
An incorrect last bit is especially problematic when all the other bits are
zero, but it turns out that in this case the correct last bit must in fact be
zero (because we assume $b \geq B(E)$, where $B$ is defined in \eqref{eq:B}
in Section~\ref{sec:B} below and can be computed via Operation~\ref{op:B}).
Operations~\ref{op:real} and~\ref{op:floor}
(which are built on Operation~\ref{op:realbits})
show that this enables exact computation of sign, floor, and ceiling,
despite these depending on the value of the last bit.
(In our quasigeodesic algorithm,
we do not actually need the floors and ceilings of Operation~\ref{op:floor},
or the approximations of Operation~\ref{op:realbits},
but we describe how to easily implement them for possible future~use.)

The expression-RAM operations of Table~\ref{tab:ops} can be implemented on
the real RAM by representing an ``expression'' as a ``built-in real number'',
except for Operation~\ref{op:floor}
(assuming no floor or ceiling operation is available, as motivated in
Section~\ref{sec:real RAM floor}).
Thus we can state a single algorithm in terms of
Operations~\ref{op:combine}, \ref{op:real}, \ref{op:B}, and \ref{op:realbits},
and then analyze the running time both on the real RAM and on the word RAM.
The rightmost column of Table~\ref{tab:ops}
gives the real-RAM cost of each operation.
Operations~\ref{op:combine}, \ref{op:real}, and~\ref{op:B}
are directly supported in one operation on the real RAM
(where Operation~\ref{op:combine} now
actually computes the real number resulting from the expression).
Operation~\ref{op:realbits} requires a bit more effort,
and will be analyzed in Lemma~\ref{lem:real RAM} below.

\begin{table}
  \centering
  \tabcolsep=6pt
  \renewcommand\arraystretch{2}
  \def\REF#1{\\[3pt]{\footnotesize [#1]}}
  \def\NOREF{\\[3pt]{\footnotesize \phantom{[]}}}
  \begin{tabular}{cm{2.6in}ccc}
    && \multicolumn{2}{c}{Word RAM Time Cost}
    \\ \cline{3-4}
    & Expression RAM Operation
    & \makecell*{Recursive\\Cost\\Model}
    & \makecell*{Simple\\Cost\\Model}
    & \makecell*{Real\\RAM\\Cost}
    \\ \hline
    \makecell{\labelthis{op:combine}{O1}. \\ \\ \\ \\}
      & Given two expressions $E_1$ and $E_2$ (possibly integer expressions)
        and a positive integer $d$, construct
        the expression $E_1 + E_2$, $E_1 - E_2$, $E_1 \cdot E_2$,
        $E_1 / E_2$, or $\sqrt[d]{E_1}$.
      & \makecell{$O(1)$ \NOREF}
      & \makecell{$O(1)$ \NOREF}
      & \makecell{$O(1)$
          \REF{Lem.~\ref{lem:real RAM}}}
    \\
    \makecell{\labelthis{op:real}{O2}. \\ \\ \\}
      & Given an expression $E$, compute the sign of $E$,
        i.e., whether it is zero, positive, or negative.
      & \makecell{$O(T(E))$ 
          \REF{Thm.~\ref{thm:real/floor}}}
      & 
        \footnotesize\makecell{$O\Big( b^*(E)^2 \cdot{}$ \\
          $\big(32\,K(E)\big)^{2 |R(E)| + 2 H(E)+6} \Big)$
          \REF{Thm.~\ref{thm:simple}}}
      & \makecell{$O(1)$
          \REF{Lem.~\ref{lem:real RAM}}}
    \\
    \makecell{\labelthis{op:B}{O3}. \\ \\ \\}
      & Given an expression $E$, compute a number $B(E)$ such that
        $E = 0$ or ${1 \over 2^{B(E)}} \leq |E| \leq 2^{B(E)}$.
      & \footnotesize\makecell{$O\Big(\sum_{E' \in E} \lg B(E')\Big)$ \\
          ${}=O\Big(\#(E) \lg B(E)\Big)$
          \REF{Lem.~\ref{lem:B}}}
      & \footnotesize\makecell{$O\Big(\#(E)^2 \lg K(E) +{}$ \\
          $\#(E) \lg b^*(E)\Big)$
          \REF{Cor.~\ref{cor:simple B}}}
      & \makecell{$O(1)$
          \REF{Lem.~\ref{lem:real RAM}}}
    \\
    \makecell{\labelthis{op:realbits}{O4}. \\ \\ \\ \\ \\ \\}
      & Given an expression $E$ and a positive integer $b \geq B(E)$,
        \defn{$\frac{1}{2^b}$-compute} $E$:
        compute an interval $[l,u]$ of rational numbers $l,u$
        (represented by quotients of integer expressions)
        such that $[l,u]$ contains the value of~$E$
        and has length $u-l \leq \frac{1}{2^b}$.
      & \makecell{$O(T(E, b))$
          \REF{Thm.~\ref{thm:realbits}}}
      & \footnotesize\makecell{$O\left(b \cdot \big(8\,K(E)\big)^{H(E)+2}\right)$
          \REF{Thm.~\ref{thm:somewhat simple}}}
      & \makecell{$O(b)$
          \REF{Lem.~\ref{lem:real RAM}}}
    \\
    \makecell{\labelthis{op:floor}{O5}. \\ \\}
      & Given an expression $E$, compute integer expression
        $\lfloor E \rfloor$ or $\lceil E \rceil$.
      & \makecell{$O\big(T(E) \cdot$ \\ $S(E) \big)$ 
          \REF{Thm.~\ref{thm:real/floor}}}
      & \footnotesize\makecell{$O\Big( b^*(E)^2 \cdot{}$ \\
          $\big(32\,K(E)\big)^{2 |R(E)| + 2 H(E)+6} \Big)$
          \REF{Thm.~\ref{thm:simple}}}
      & n/a
    \\
    \multicolumn{1}{r}{$\bullet$}
      & \multicolumn{4}{p{6in}}{
          In Operations \ref{op:real}, \ref{op:realbits}, and \ref{op:floor},
          if $E$ contains any invalid computation ---
          division by zero, or even roots of negative numbers
          ---
          then these computations simply produce a special result of ``undefined''
          (following \cite{SepBound}).
        }
    \\
  \end{tabular}

  \caption{Expression operations supported by the expression RAM,
    and their costs on the word RAM (middle columns)
    and on the real RAM (right column).
    (In addition, the expression RAM supports all word RAM operations.)
    Section~\ref{sec:B} defines $B(E)$ and $D(E)$;
    Section~\ref{sec:recursive cost model} defines $S(E)$, $T(E)$, and $T(E,b)$; and
    Section~\ref{sec:simple cost model} defines
    $\#(E)$, $b^*(E)$, $K(E)$, $R(E)$, and $H(E)$.}
  \label{tab:ops}
\end{table}

In the remainder of this section, we define and prove the various time bounds
in Table~\ref{tab:ops}, including two different cost models on the word RAM.
In particular, we develop algorithms for the expression-RAM operations
on the real RAM and word RAM, and analyze these algorithms
to prove the claimed time bounds in the underlying models.
Our ``recursive'' cost model (Section~\ref{sec:recursive cost model})
specifies the running time of expression-RAM operations on the word RAM
in terms of a recurrence over the structure of the expression.
This cost model is more difficult to work with,
but may be of use in certain situations
where the expressions have specific form.
It will also serve as the foundation
for our ``simple'' cost model (Section~\ref{sec:simple cost model}),
which bounds the running time of expression-RAM operations on the word RAM
in terms of parameters that are more natural and easier to bound.
This model also includes algorithms to remove duplicate roots in an expression,
so that the running time can depend on the number of distinct roots.

It is unlikely that trigonometric functions could be added
to the expression RAM, as it is unclear how to bound their separation
from zero in general and thus obtain word-RAM algorithms for
Operation~\ref{op:real}.
For example, it is undecidable to determine whether a single-variable function
built from operators $+$, $-$, $\times$, $\sin$, $\exp$, and composition
is always nonnegative \cite{Laczkovich-2003}.


\subsubsection{Expression Bounds: $B(E)$, $C(E)$, and $D(E)$}

\label{sec:B}
At the center of our approach is a ``separation bound'' from \cite{SepBound}
that limits how close to zero an expression $E$ can be without
actually being zero.
We express this bound in terms of three functions from expressions to
positive integers,
which are simplifications of functions from \cite[Theorem~1]{SepBound}.
Specifically, define the functions $B(E)$ (``bit bound''),
$C(E)$ (``calculation complexity''), and $D(E)$ (``degree disadvantage'') as follows:
\begin{align}
\label{eq:B}
B(E) &= C(E) \cdot D(E),
\\
\label{eq:C}
C(E) &= \begin{cases}
  \max\{1,\lceil \lg |E| \rceil\} & \text{if $E$ is an integer expression,}
  \\
  2 [C(E_1) + C(E_2)]
      & \text{if $E = E_1 \circ E_2$ for some operator $\circ \in \{+,-,\times,\div\}$,}
  \\
  2 \, C(E_1)
      & \text{if $E = \sqrt[d]{E_1}$,}
\end{cases}
\end{align}
\begin{align}
%
%
  D(E)
  &= \text{the product of the degrees $d$ of all
        radical expressions $\sqrt[d]{E_1}$ in~$E$} \label{eq:D}
  \\
  &\leq \begin{cases}
    1 & \text{if $E$ is an integer expression,}
    \\
    D(E_1) \cdot D(E_2)
        & \text{if $E = E_1 \circ E_2$ for some operator $\circ \in \{+,-,\times,\div\}$,}
    \\
    d \cdot D(E_1)
        & \text{if $E = \sqrt[d]{E_1}$.}
    \label{ineq:D}
  \end{cases}
\end{align}
(Inequality \eqref{ineq:D} upper bounding $D(E)$ is an equality for trees,
but can over-count for DAGs.)

Crucially, these functions provide bounds on how large an expression $E$
can get, and on how close to zero a nonzero expression $E$ can get:

\begin{theorem} {\rm \cite[Theorem~1]{SepBound}} \label{thm:sep}
  Any real algebraic expression $E \neq 0$ satisfies
  \begin{equation} \label{eq:CD bounds}
    {1 \over 2^{B(E)}} \leq |E| \leq 2^{B(E)},
  \end{equation}
  or taking logarithms,
  \begin{equation} \label{eq:B bound}
    -B(E) \leq \lg |E| \leq B(E).
  \end{equation}
\end{theorem}

\begin{proof}
  In fact, \cite[Theorem~1]{SepBound} states that
  \begin{equation} \label{eq:SepBound}
    {1 \over l(E) \cdot u(E)^{D(E)-1}} \leq |E| \leq u(E) \cdot l(E)^{D(E)-1},
  \end{equation}
  where $D(E)$ is defined the same, and $l(E)$ and $u(E)$ are
  defined recursively as in the cases below.%
  \footnote{The functions $l(E)$ and $u(E)$ from \cite{SepBound}
    are used only within our proof of Theorem~\ref{thm:sep}, and
    are unrelated to the rational numbers $l$ and $u$ output by
    Operation~\ref{op:realbits}.}
  We have simplified the statement of the theorem to give a
  rough upper bound $2^C$ that applies in all cases, related to their functions
  $u,l$ via $2^{C(E)} \geq u(E) \cdot l(E)$.
  The base case (for integer expressions~$E$) has been modified
  to guarantee $C(E) \geq 1$,
  which allows us to simplify some of the recursive formulas.
  We have also loosened the bound by exponentiating both $l(E)$ and
  $u(E)$ by $D(E)$ (which is also $1$ larger than needed).

  We prove by structural induction on the expression $E$ that
  $2^{C(E)} \geq u(E) \cdot l(E)$.
  First note that $2^{C(E)} \ge 2$ and $u(E),l(E) \ge 1$ always.
  Then we proceed in cases:
\begin{enumerate}
\item
If $E$ is an integer expression $N$, 
then $u(E) = |N|$ and $l(E) = 1$, 
so $u(E) \cdot l(E) = |N| \le \max\{2,2^{\lceil \lg |E|\rceil}\} = 2^{C(E)}$.
\item
If $E = E_1 \circ E_2$ for some operator $\circ \in \{+,-\}$,
then $u(E) = u(E_1) \cdot l(E_2) + l(E_1) \cdot u(E_2) \le 2\max\{u(E_1) \cdot l(E_2), l(E_1) \cdot u(E_2)\}$ and $l(E) = l(E_1) \cdot l(E_2)$,
so
\begin{align*}
  u(E) \cdot l(E)
  &\leq 2\max\{u(E_1) \cdot l(E_1) \cdot l(E_2)^2, l(E_1)^2 \cdot l(E_2) \cdot u(E_2)\}
  \\
  &\leq 2 \, u(E_1) \cdot l(E_1) \cdot u(E_2) \cdot l(E_2) \cdot \max\{u(E_2)  \cdot l(E_2), u(E_1)  \cdot l(E_1)\}
  \\
  &\leq 2 \cdot 2^{C(E_1)} \cdot 2^{C(E_2)} \cdot \max\{2^{C(E_2)}, 2^{C(E_1)}\},
  \\
  &= 2^{1+C(E_1)+C(E_2)+\max\{C(E_2), C(E_1)\}}.
\end{align*}
Because $\min\{C(E_1),C(E_2)\} \ge 1$,
we have that $u(E) \cdot l(E)
\le 2^{2 C(E_1) + 2 C(E_2)} = 2^{C(E)}$, as claimed.
\item
If $E = E_1 \cdot E_2$,
then $u(E) = u(E_1) \cdot u(E_2)$ and $l(E) = l(E_1) \cdot l(E_2)$, so
$$u(E) \cdot l(E)
= u(E_1) \cdot l(E_1) \cdot l(E_2) \cdot u(E_2)
\le 2^{C(E_1)} \cdot 2^{C(E_2)}
= 2^{C(E_1) + C(E_2)}
= 2^{C(E)},$$
as claimed.
\item
If $E = E_1 / E_2$,
then $u(E) = u(E_1) \cdot l(E_2)$ and $l(E) = l(E_1) \cdot u(E_2)$, so
$$u(E) \cdot l(E)
= u(E_1) \cdot l(E_1) \cdot l(E_2) \cdot u(E_2)
\le 2^{C(E_1)} \cdot 2^{C(E_2)}
= 2^{C(E_1) + C(E_2)}
= 2^{C(E)},$$
as claimed.
\item
If $E = \sqrt[d]{E_1}$ and $u(E_1) \ge l(E_1)$,
then $u(E) = u(E_1)^{\frac 1 d} \cdot l(E_1)^{\frac{d-1}d}$
and $l(E) = l(E_1)$,
so $u(E) \cdot l(E)
= u(E_1)^{\frac 1 d} \cdot l(E_1)^{\frac{2d-1}d}
\le \left(2^{C(E_1)}\right)^2
= 2^{2 C(E_1)} = 2^{C(E)}$, as claimed.
\item
If $E = \sqrt[d]{E_1}$ and $u(E_1) < l(E_1)$,
then $u(E) = u(E_1)$ and $l(E) = u(E_1)^{\frac{d-1}d} \cdot l(E_1)^{\frac{1}d}$,
so $u(E) \cdot l(E)
= u(E_1)^{\frac{2d-1}d} \cdot l(E_1)^{\frac{1}d}
\le \left(2^{C(E_1)}\right)^2
= 2^{2 C(E_1)} = 2^{C(E)}$, as claimed.
\end{enumerate}

Finally, (\ref{eq:SepBound}) shows that $|E| \le u(E) \cdot l(E)^{D(E)-1}$,
which is at most $(u(E) \cdot l(E))^{D(E)} \leq 2^{C(E) \cdot D(E)}$;
and that $|E| \ge u(E)^{1-D(E)} \cdot l(E)^{-1}$,
which is at least $(u(E) \cdot l(E))^{-D(E)} \geq 2^{- C(E) \cdot D(E)}$,
as desired.
\end{proof}

\subsubsection{Real RAM Costs}

Next we show how to implement the expression-RAM operations
(except Operation~\ref{op:floor}) efficiently on the real RAM.



\begin{lemma} \label{lem:real RAM}
  Operations~\ref{op:combine},~\ref{op:real},~\ref{op:B}, and~\ref{op:realbits}
  can be implemented on the real RAM in the time bounds stated in
  Table~\ref{tab:ops}.
\end{lemma}


\begin{proof}
  Operation~\ref{op:combine} constructs the DAG node
  and actually performs the real-number computation,
  making Operation~\ref{op:real} run in constant time.

  Operation~\ref{op:B} computes $C(E)$ and an upper bound $D'(E)$ on $D(E)$
  using $O(1)$ operations on those values
  from the input expressions to implement recursions
  \eqref{eq:C} and \eqref{ineq:D};
  then we use $B(E) = C(E) \cdot D'(E)$.

  Finally, to perform Operation~\ref{op:realbits} in $O(b)$ time,
  we first compute $B(E)$ using Operation~\ref{op:B} in $O(1)$ time.
  Start with the interval $[-2^{B(E)}, 2^{B(E)}]$
  (where $B(E)$ is given to us by the last Operation~\ref{op:B}),
  which by Theorem~\ref{thm:sep} contains the given real number~$E$.
  We perform a binary search on the interval $[l, u]$.
  At each step, we compute the midpoint $m = (l+u)/2$
  via Operation~\ref{op:combine},
  compare $m$ against $E$ via Operation~\ref{op:real},
  and set $l$ or $u$ to $m$ according to whether $m \geq E$ or $m \leq E$.
  Each step takes $O(1)$ time, preserves $E \in [l, u]$,
  and reduces the interval length $u-l$ by a factor of~$2$.
  After $B(E) \leq b$ steps,
  the interval length is $\leq 2$.
  After another $b+1$ steps, the interval length is $\leq 1/2^b$.
  The total number of steps and thus running time is $O(b)$.
\end{proof}

\subsubsection{Recursive Cost Model}
\label{sec:recursive cost model}

In the recursive cost model, the expression-RAM operations have the
following time costs on the word RAM.
Operation~\ref{op:realbits} runs in $O(T(E,b))$ time,
Operation~\ref{op:real} runs in $O(T(E))$ time,
and Operation~\ref{op:floor} runs in $O(T(E)^2)$ time,
where $T(E) = T(E, B(E))$ and
$T(E,b)$ is given by the following recurrence:
\begin{equation} \label{eq:T}
T(E,b) = \begin{cases}
  S(E,b) & \text{\hspace*{-5em}if $E$ is an integer expression,}
  \\
    T(E_1, b + 1) + T(E_2, b + 1) + S(E,b)
      & \text{if $E = E_1 \pm E_2$,}
  \\
    T(E_1, b + B(E)) + T(E_2, b + B(E)) + S(E,b)
      & \text{if $E = E_1 \cdot E_2$,}
  \\
    T(E_1, b + B(E)) + T(E_2, b + 2 B(E) + 2) + S(E,b)
      & \text{if $E = E_1 / E_2$,}
  \\
  T(E_1, 2 d b) + d \cdot S(E,b)
      & \text{if $E = \sqrt[d]{E_1}$.}
\end{cases}
\end{equation}
%
Here $S(E,b)$ represents the maximum number of bits in the numerator or
denominator of $l$ or $u$ in the interval $[l,u]$ output by
Operation~\ref{op:realbits} on input $E$ and $b$;
$S(E,b)$ is given by another recurrence,
with the same recursive structure as $T(E,b)$
but differing in the additive terms:
\begin{equation} \label{eq:S}
S(E,b) = \begin{cases}
  B(E) = C(E) = b' & \text{\hspace*{-5em}if $E$ is a $b'$-bit integer expression,}
  \\
    S(E_1, b + 1) + S(E_2, b + 1) + 1
      & \text{if $E = E_1 \pm E_2$,}
  \\
    S(E_1, b + B(E)) + S(E_2, b + B(E))
      & \text{if $E = E_1 \cdot E_2$,}
  \\
    S(E_1, b + B(E)) + S(E_2, b + 2 B(E) + 2)
      & \text{if $E = E_1 / E_2$,}
  \\
  S(E_1, 2 d b) + 3 d b 
      & \text{if $E = \sqrt[d]{E_1}$.}
\end{cases}
\end{equation}

To give some intuition about how $S$ and $T$ behave,
we prove that their dependence on $b$ is at most linear:

\begin{lemma} \label{lem:T affine}
  For a fixed expression~$E$,
  $T(E,b)$ is an affine function $c \cdot b + c'$ of $b$,
  where both $c$ and $c'$ are nonnegative integers.
  The same is true for $S(E,b)$.
\end{lemma}

\begin{proof}
  Let $b_0$ represent the parameter $b$ at the top level of the recurrence.
  The recurrences for $T(E,b)$ and $S(E,b)$ are all sums of terms.
  Expanding out all occurrences of $T$ and $S$, we obtain a sum of terms
  all of the form $1$, $b'$, and $3 d b$ for positive integers $b',d$
  (dependent only on~$E$).
  In addition, some terms are multiplied by $d$
  (from the last case of $T(E,b)$).
  The $3 d b$ terms use a derived value of $b$, which is formed from $b_0$
  by repeatedly adding $1$, adding $B(E)$, adding $2 B(E)$,
  or multiplying by $2 d$.
  By induction, every $b$ value is an affine function of $b_0$
  with nonnegative integer coefficients, and thus so is each term $3 d b$.
  Therefore the sum of terms resulting from expanding $T(E,b_0)$ or $S(E,b_0)$
  is also such an affine function.
\end{proof}

\begin{corollary} \label{cor:growth in b}
  For any constant $c \geq 1$, $T(E,c \cdot b) \leq c \cdot T(E,b)$.
  The same is true for $S(E,b)$.
\end{corollary}

We also analyze how quickly $T(E)$ grows when $E$ gains an operation:

\begin{lemma} \label{lem:T minus}
  For two expression DAGs $E_1$ and $E_2$,
  $T(E_1 \pm E_2) = O\big(T(E_1) \cdot T(E_2)\big)$.
\end{lemma}

\begin{proof}
  First we show that $B(E_1 \pm E_2) = O\big(B(E_1) \cdot B(E_2)\big)$:
  \begin{align*}
  D(E_1 \pm E_2) &= D(E_1) \cdot D(E_2) && \text{by \eqref{eq:D}}
  \\
  C(E_1 \pm E_2) &= 2 \, [C(E_1) + C(E_2)] && \text{by \eqref{eq:C}}
  \\
  B(E_1 \pm E_2)
  &= C(E_1 \pm E_2) \cdot D(E_1 \pm E_2)
  \\
  &= 2 \, [C(E_1) + C(E_2)] \cdot D(E_1) \cdot D(E_2)
  \\
  &\leq 2 \, [2 \, C(E_1) \cdot C(E_2)] \cdot D(E_1) \cdot D(E_2) && \text{because $C(E_i) \geq 1$}
  \\
  &= 4 \, B(E_1) \cdot B(E_2).
  \end{align*}
  Define $b = B(E_1 \pm E_2) \leq 4 \, B(E_1) \cdot B(E_2)$.
  Now we compute $T(E_1 \pm E_2) = T(E_1 \pm E_2,b)$ as follows:
  $$
  \arraycolsep=0pt
  \def\arraystretch{1.2}
  \begin{array}{ccccccccccll}
  T(E_1 \pm E_2, b)
  &{}={}& T(E_1, b+1) &{}+{}& T(E_2, b+1) &{}+{}& \multicolumn{3}{l}{S(E_1 \pm E_2,b)} &
  &{\quad}& \text{by \eqref{eq:T}}
  \\
  &{}={}& T(E_1, b+1) &{}+{}& T(E_2, b+1) &{}+{}& S(E_1,b+1) &{}+{}& S(E_2,b+1) &{}+ 1
  && \text{by \eqref{eq:S}}
  \\
  &{}\leq{}& 2 \cdot T(E_1,b) &{}+{}& 2 \cdot T(E_2,b) &{}+{}& 2 \cdot S(E_1,b) &{}+{}& 2 \cdot S(E_2,b) &{}+ 1
  && \text{by Corollary~\ref{cor:growth in b} and $b \geq 1$}
  \\
  &{}\leq{}& 2 \cdot T(E_1,b) &{}+{}& 2 \cdot T(E_2,b) &{}+{}& 2 \cdot T(E_1,b) &{}+{}& 2 \cdot T(E_2,b) &{}+ 1,
  \end{array}$$
  where the last inequality follows from
  $T(E,b) \geq S(E,b)$ by inspection of~\eqref{eq:T}.

  Therefore $T(E_1 \pm E_2) = T(E_1 \pm E_2,b) = O(T(E_1,b) + T(E_2,b))$.
  By Lemma~\ref{lem:T affine},
  $T(E_i,b) = c_i \cdot b + c'_i$ for nonnegative integers $c_i,c'_i$
  (which depend on~$E_i$).
  Thus we can bound the sum
  %
  \begin{align*}
  T(E_1,b) + T(E_2,b)
  &= (c_1 \cdot b + c'_1) + (c_2 \cdot b + c'_2)
  \\
  &= (c_1 + c_2) \cdot b + (c'_1 + c'_2)
  \\
  &\leq (c_1 + c_2) \cdot \big(4 \cdot B(E_1) \cdot B(E_2)\big) + (c'_1 + c'_2)
  \\
  &\leq 4 \cdot \big(c_1 \cdot B(E_1) + c'_1\big) \cdot \big(c_2 \cdot B(E_2) + c'_2\big)
  \\
  &= 4 \cdot T\big(E_1,B(E_1)\big) \cdot T\big(E_2,B(E_2)\big)
  \\
  &= 4 \cdot T(E_1) \cdot T(E_2).
  \end{align*}
  In conclusion, $T(E_1 \pm E_2) = O\big(T(E_1) \cdot T(E_2)\big)$
  as desired.
\end{proof}

We also give a weak relation between $T(E)$ and $B(E)$:

\begin{lemma} \label{lem:S vs B}
  $\lg B(E) \leq 2 \, S(E)$.
\end{lemma}

\begin{proof}
  First we upper bound $C(E) \leq C'(E)$
  by increasing the base case of \eqref{eq:C} so that $C'(E) \geq 2$:
  \begin{align*}
  C'(E)
  &= \begin{cases}
    \max\{2,b'\} & \text{if $E$ is a $b'$-bit integer expression,}
    \\
    2 [C'(E_1) + C'(E_2)]
        & \text{if $E = E_1 \circ E_2$ for some operator $\circ \in \{+,-,\times,\div\}$,}
    \\
    2 \, C'(E_1)
        & \text{if $E = \sqrt[d]{E_1}$.}
  \end{cases}
  \end{align*}
  Note that $C'(E) \geq 4$ for any nonleaf expression $E$.
  Taking logarithms,
  \begin{align*}
  \lg C'(E)
  &= \begin{cases}
    \max\{1,\lg b'\} & \text{if $E$ is a $b'$-bit integer expression,}
    \\
    \lg [C'(E_1) + C'(E_2)] + 1
        & \text{if $E = E_1 \circ E_2$ for some operator $\circ \in \{+,-,\times,\div\}$,}
    \\
    \lg C'(E_1) + 1
        & \text{if $E = \sqrt[d]{E_1}$.}
  \end{cases}
  \end{align*}
  For $x,y \geq 2$, $\lg$ is subadditive:
  $\lg (x+y) \leq \lg x + \lg y$.
  In fact, for $x,y \geq 4$,
  $\lg (x+y) + 1 \leq \lg x + \lg y$.
  Thus
  $\lg [C'(E_1) + C'(E_2)] + 1 \leq \lg C'(E_1) + \lg C'(E_2)$
  if $E_1,E_2$ are nonleaf expressions,
  whereas we need the $+ 1$ if $E_1$ or $E_2$ is a leaf expression.
  By amortization, we can account for each needed $+ 1$ by charging
  to one of the child leaves, and thus we get an upper bound of
  $\lg C'(E) \leq X(E)$ where $X(E)$ moves the $+1$ to the leaf case:
  \begin{align*}
  X(E)
  &= \begin{cases}
    \max\{1,\lg b'\} + 1 & \text{if $E$ is a $b'$-bit integer expression,}
    \\
    X(E_1) + X(E_2)
        & \text{if $E = E_1 \circ E_2$ for some operator $\circ \in \{+,-,\times,\div\}$,}
    \\
    X(E_1) + 1
        & \text{if $E = \sqrt[d]{E_1}$.}
  \end{cases}
  \end{align*}
  By \eqref{eq:B} and \eqref{eq:D},
  we can similarly upper bound $\lg B(E) = \lg C(E) + \lg D(E) \leq Y(E)$
  by adding $d$ to the radical case:
  \begin{align*}
  Y(E)
  &= \begin{cases}
    \max\{1,\lg b'\} + 1 & \text{if $E$ is a $b'$-bit integer expression,}
    \\
    X(E_1) + X(E_2)
        & \text{if $E = E_1 \circ E_2$ for some operator $\circ \in \{+,-,\times,\div\}$,}
    \\
    X(E_1) + d+1
        & \text{if $E = \sqrt[d]{E_1}$.}
  \end{cases}
  \end{align*}
  By \eqref{eq:S}, $Y(E) \leq 2 \, S(E,b)$:
  in particular, $\max\{1,\lg b'\} + 1 \leq 2 b'$ for all $b' \geq 1$
  and $d+1 \leq 2 \cdot 3 d b$ for all $b, d \geq 1$.
  Thus we obtain the desired upper bound $\lg B(E) \leq Y(E) \leq 2 \, S(E,b)$.
\end{proof}

Now that we have stated the target time bounds, we develop algorithms
for the expression-RAM operations and prove that they achieve these bounds.
First we need a lemma for working with arbitrary-precision integers
and rationals.

\begin{lemma} \label{lem:rationals}
  In $O(b)$ time on the word RAM,
  we can add, subtract, multiply, integer divide, and compare
  $b$-bit integers; and we can add, subtract, multiply, divide, and compare
  rationals represented as quotients of $b$-bit integers.
\end{lemma}

\begin{proof}
  Integer addition and subtraction can be implemented
  using the grade-school algorithms, working in base $2^w$,
  with a running time of $O(b/w) = O(b)$ time.
  Integer comparison follows from subtraction and checking the resulting sign.
  An $O(b)$-time algorithm for integer multiplication on a word RAM with
  $w = \Omega(\lg b)$ goes back to Sch\"onhage in 1980
  \cite{Schoenhage-1980}; see also \cite{Fuerer-2014} and
  \cite[\S 4.3.3.C, p.~311 and Exercise 4.3.3-12, p.~317]{TAOCP2}.
  Furthermore, integer division can be implemented by a series of
  integer multiplications of geometrically decreasing size
  \cite[\S 4.3.3.D, pp.~311--313]{TAOCP2},
  so also in $O(b)$ time on a word RAM.

  Each rational arithmetic operation reduces to $O(1)$ integer arithmetic
  operations via the grade-school identities:
  $$
    \frac{p}{q} + \frac{r}{s} = \frac{p \cdot s + r \cdot q}{q \cdot s};
    \quad
    \frac{p}{q} - \frac{r}{s} = \frac{p \cdot s - r \cdot q}{q \cdot s};
    \quad
    \frac{p}{q} \cdot \frac{r}{s} = \frac{p \cdot r}{q \cdot s};
    \quad
    \left. \frac{p}{q} \right/ \frac{r}{s} = \frac{p \cdot s}{q \cdot r}.
  $$
  The sign of a rational number is the product of signs of its numerator and
  denominator, so we can support comparisons via subtraction.
  (We do not bother to reduce fractions,
  as that does not help in the worst case.)
\end{proof}

As the proof shows, the cost of integer addition, subtraction, and comparison
can be tightened to $O(b/w)$ time.
Recently, $b$-bit integer multiplication was shown to be possible in
$O(b \lg b)$ bit operations \cite{integer-multiply}, but to our knowledge
it remains open whether it is possible to achieve $o(b)$ time
(ideally, $O(b/w)$ time) on the word RAM.
Thus we opt for the simpler universal upper bound of $O(b)$ for all operations,
though perhaps this could be lowered in the future.

Next we implement the simplest expression-RAM operations,
Operation~\ref{op:combine} and~\ref{op:B}.
Operation~\ref{op:combine} just constructs the DAG node, so takes constant time.

\begin{lemma} \label{lem:B}
  Operation~\ref{op:B} can be implemented in
  $O\left(\sum_{E' \in E} \lg B(E')\right) = O\big(\#(E) \lg B(E)\big)$ time,
  where $E' \in E$ denotes $E'$ being a node of the expression tree~$E$,
  and $\#(E)$ is the number of nodes in expression~$E$.
  In the same time, we can compute $B(E')$, $C(E')$, and $D(E')$
  for all $E' \in E$.
\end{lemma}

\begin{proof}
  We perform a postorder traversal of the given expression DAG $E$.
  At a node
  $E' \in \{E_1 + E_2, E_1 - E_2, E_1 \cdot E_2, E_1 / E_2, \sqrt[d]{E_1}\}$,
  we compute $C(E')$ using the already computed values $C(E_i)$
  via the recurrence (\ref{eq:C}).
  This recurrence involves $O(1)$ additions and multiplications
  on $O(\lg B(E'))$-bit integers, so by Lemma~\ref{lem:rationals},
  takes $O(\lg B(E'))$ time.
  To compute $D(E')$, we maintain a running product of visited radical nodes.
  Each product is between a number at most $D(E')$ and a number at most
  $K(E') \leq D(E')$,
  so by Lemma~\ref{lem:rationals}, takes $O(\lg D(E')) = O(\lg B(E'))$ time.
  We can then compute $B(E') = C(E') \cdot D(E')$ in $O(\lg B(E'))$ time
  by Lemma~\ref{lem:rationals}.
  In total over the $\#(E)$ nodes,
  the postorder traversal costs
  $O\left(\sum_{E' \in E} \lg B(E')\right) = O\big(\#(E) \lg B(E)\big)$ time.
\end{proof}

Now we implement the core expression-RAM operation, Operation~\ref{op:realbits}.
Operations~\ref{op:real} and~\ref{op:floor} will be implemented on top of this
operation.

\begin{theorem} \label{thm:realbits}
  Operation~\ref{op:realbits} can be implemented on a word RAM in
  $O(T(E,b))$ time.
\end{theorem}


\begin{proof}
  We show how to $\epsilon$-compute $E$ where $\epsilon = \frac{1}{2^b}$
  and $b \geq B(E) \geq 1$ recursively in cases.
  Roughly speaking, each expression asks its children expressions
  for increasing precision $b'$ so that, when we combine the intervals
  from the children, we meet the desired interval length bound at the parent.
  Along the way, we keep track of how many bits $S$ are needed to represent the
  intervals themselves.
  We apply standard identities for interval arithmetic, specifically
  addition, subtraction, and multiplication; see
  \cite{Gibb-1961} and \cite[\S 2.2]{Moore-1979}.

  More formally, we prove by induction on the number of nodes in input
  expression $E$ that Operation~\ref{op:realbits}
  (1)~runs in $O(T(E,b))$ time, and
  (2)~produces an interval of rational numbers whose numerators
  and denominators each have at most $S(E,b)$ bits (ignoring the sign bit).
  Assume by induction that smaller expressions satisfy these
  two induction hypotheses.

  First, at the top level (not recursively),
  we apply Lemma~\ref{lem:B} to compute $B(E')$,
  for all nodes $E' \in E$, in $O\left(\sum_{E' \in E} \lg B(E')\right)$ time.
  By Lemma~\ref{lem:S vs B},
  this cost can be absorbed by charging $\lg B(E')$ to the $S(E',b')$
  we already pay at each node $E' \in E$
  in the $T(E,b)$ recurrence \eqref{eq:T}.


  \paragraph{Addition and subtraction.}
  To $\epsilon$-compute $E = E_1 \circ E_2$ where $\circ \in \{+,-\}$,
  we first recursively $\epsilon'$-compute $E_1$ and $E_2$
  where $\epsilon' = \epsilon/2$.
  In other words, the recursive call is with $b'=b+1$.
  Call the resulting intervals $[l_1,u_1]$ and $[l_2,u_2]$.
  Then $[l_1 \circ l_2,u_1 \circ u_2]$ is an $\epsilon$-computation of
  $E=E_1 \circ E_2$:
  \begin{align*}
  u_1 \circ u_2 - (l_1 \circ l_2)
  &= (u_1 - l_1) \circ (u_2 - l_2)
  \\
  &\leq \epsilon' + \epsilon' \quad \text{(by triangle inequality when $\circ$ is $-$)}
  \\
  &\leq 2 \epsilon' = \epsilon.
  \end{align*}
  We can compute the interval $[l_1 \circ l_2,u_1 \circ u_2]$ of rationals
  using Lemma~\ref{lem:rationals}.
  By induction, the numerators and denominators of $l_i$ and $u_i$ have
  at most $S(E_i,b')$ bits, so the numerators and denominators of the
  sums/differences $l_1 \circ l_2$ and $u_1 \circ u_2$ have at most
  \begin{align*}
    \max\{ S(E_1,b'), S(E_2,b') \} + 1
    &\leq S(E_1,b') + S(E_2,b') + 1
    \\
    &= S(E_1,b+1) + S(E_2,b+1) + 1
    \\
    &= S(E,b)
  \end{align*}
  bits, and computing them takes $O(S(E,b))$ time beyond the recursive calls.

  \paragraph{Multiplication.}
  To $\epsilon$-compute $E = E_1 \cdot E_2$,
  we first recursively $\epsilon'$-compute $E_1$ and $E_2$
  where $\epsilon' = \epsilon/2^{B(E)}$.
  In other words, the recursive call is with $b' = b + B(E)$,
  where $B(E)$ was already computed above.
  Call the resulting intervals $[l_1,u_1]$ and $[l_2,u_2]$.
  To properly handle signs, we can take all pairwise products
  and take their min and max \cite{Gibb-1961,Moore-1979}:
  $[\min\{l_1 \cdot l_2, l_1 \cdot u_2, u_1 \cdot l_2, u_1 \cdot u_2\},
    \max\{l_1 \cdot l_2, l_1 \cdot u_2, u_1 \cdot l_2, u_1 \cdot u_2\}]$
  contains the value of $E = E_1 \cdot E_2$.
  In fact, we claim it is an $\epsilon$-computation of $E$.

  There are $16 = 4^2$ cases for the choices of max and min.
  To enumerate the cases, we introduce some notation:
  for $x \in \{l, u\}$, define $\overline x$ is the other of $l$ and~$u$,
  i.e., $\overline x \in \{l, u\} \setminus \{x\}$;
  and for $i \in \{1, 2\}$, define
  $\overline i = 3 - i \in \{1, 2\} \setminus \{i\}$.
  Note that all terms in the min and max are of the form $x_1 \cdot y_2$,
  or by symmetry and commutativity, of the form $x_i \cdot y_{\overline i}$.
  We split the cases into groups according to how many of
  $l_1, l_2, u_1, u_2$ appear.
  \begin{enumerate}
  \item Exactly two of $l_1, l_2, u_1, u_2$ appear.
  These four of the sixteen cases have the form
  $[x_1 y_2, x_1 y_2]$ where $x,y \in \{l,u\}$.
  Then the max equals the min, so we get a $0$-computation.
  \item Exactly three of $l_1, l_2, u_1, u_2$ appear.
  These eight of the sixteen cases have the form
  $[x_i \cdot y_{\bar i}, x_i \cdot \bar y_{\bar i}]$.
  where $x,y \in \{l,u\}$ and $i \in \{1,2\}$.
  (Only four of these cases have positive length.)
  Then the length of the interval can be bounded as follows:
  \begin{align*}
  |x_i \cdot y_{\bar i} - x_i \cdot \bar y_{\bar i}|
  &= |x_i \cdot (y_{\bar i} - \bar y_{\bar i})|
  \\
  &\leq |x_i| \cdot |y_{\bar i} - \bar y_{\bar i}|
  \\
  &= |x_i| \cdot (u_{\bar i} - l_{\bar i})
  \\
  &\le \left( |E_i| + \epsilon' \right) \cdot \epsilon'
  && \text{because $E_i \in [l_i, u_i]$}
  \\
  &\le \left( 2^{B(E_i)} + \epsilon' \right) \cdot \epsilon'
  && \text{by Theorem~\ref{thm:sep}}
  \\
  &\le 2 \cdot 2^{B(E_i)} \cdot \epsilon'
  && \text{because $\epsilon' \leq 1 \leq 2^{B(E_i)}$}
  \\
  &\le 2^{B(E)} \cdot \epsilon'
  && \text{because $D(E) \geq D(E_i)$ and $C(E) \geq C(E_i)+1$}
  \\
  &= \epsilon.
  \end{align*}
  \item All four of $l_1, l_2, u_1, u_2$ appear.
  These four of the sixteen cases have the form
  $[x_1 \cdot y_2, \bar x_1 \cdot \bar y_2]$.
  where $x,y \in \{l,u\}$.
  (Only two of these cases have positive length.)
  Then the length of the interval can be bounded as follows:
  \begin{align*}
  |x_1 \cdot y_2 - \bar x_1 \cdot \bar y_2|
  &= |x_1 \cdot y_2 - x_1 \cdot \bar y_2 + x_1 \cdot \bar y_2 - \bar x_1 \cdot \bar y_2|
  \\
  &= |x_1 \cdot (y_2 - \bar y_2) + (x_1 - \bar x_1) \cdot \bar y_2|
  \\
  &\leq |x_1| \cdot |y_2 - \bar y_2| + |x_1 - \bar x_1| \cdot |\bar y_2|
  \\
  &= |x_1| \cdot (u_2 - l_2) + (u_1 - l_1) \cdot |\bar y_2|
  \\
  &\leq |x_1| \cdot \epsilon' + \epsilon' \cdot |\bar y_2|
  \\
  &= \epsilon' \cdot (|x_1| + |\bar y_2|)
  \\
  &\leq \epsilon' \cdot \left(|E_1| + \epsilon' + |E_2| + \epsilon'\right)
  && \text{because $E_i \in [l_i, u_i]$}
  \\
  &\leq \epsilon' \cdot \left(2 \cdot 2^{B(E_1)} + 2 \cdot 2^{B(E_2)}\right)
  && \text{because $\epsilon' \leq 1 \leq 2^{B(E_i)}$}
  \\
  &\leq \epsilon' \cdot \left(2 \cdot 2^{B(E_1)} \cdot 2^{B(E_2)}\right)
  && \text{because $2^{B(E_i)} \geq 2$}
  \\
  &\leq \epsilon' \cdot 2^{B(E)}
  && \text{because $D(E) \geq D(E_i)$}
  \\
  &&& \text{and $C(E) \geq C(E_1)+C(E_2)+1$}
  \\
  &= \epsilon.
  \end{align*}

  \end{enumerate}

  We can compute the interval in $O(S(E,b))$ time using $O(1)$ rational
  multiplications and comparisons via Lemma~\ref{lem:rationals}:
  by induction, the numerators and denominators of $l_i$ and $u_i$ have
  at most $S(E_i,b')$ bits, so the number of bits in the numerators and
  denominators of the products
  $l_1 \cdot l_2, l_1 \cdot u_2, u_1 \cdot l_2, u_1 \cdot u_2$ is at most
  \begin{align*}
    S(E_1,b') + S(E_2,b')
    &= S(E_1, b + B(E)) + S(E_2, b + B(E))
    \\
    &= S(E,b).
  \end{align*}

  \paragraph{Division.}
  To $\epsilon$-compute $E = E_1 / E_2$, we reduce to multiplying
  $E_1$ and $1/E_2$, as just analyzed, which requires
  $\epsilon'$-computing $E_1$ and $1/E_2$
  where $\epsilon' = \epsilon/2^{B(E)}$,
  i.e., computing $E_1$ and $1/E_2$ with $b' = b + B(E)$.

  To $\epsilon'$-compute $E'_2 = 1/E_2$,
  we first recursively $\epsilon_2$-compute $E_2$
  where $\epsilon_2 = \epsilon/2^{B(E'_2)+2}$.
  In other words, the recursive call is with $b_2 = b' + B(E'_2) + 2$,
  where $B(E'_2)$ was already computed above.
  Call the resulting interval $[l_2,u_2]$.
  If $l_2 \leq 0 \leq u_2$, then
  $$
  |E_2|
  \leq u_2 - l_2
  \leq \frac{1}{2^{b_2}}
  \leq \frac{1}{2^b \cdot 2^{B(E'_2)}}
  < \frac{1}{2^{B(E'_2)}} \quad \text{(because $b \geq 1$)},
  $$
  so $E_2 = 0$ by Theorem~\ref{thm:sep};
  thus we return ``undefined'' to indicate division by zero.
  Otherwise, $[1/u_2, 1/l_2]$ is an $\epsilon'$-computation of $E'_2 = 1/E_2$:
  \begin{align*}
    \frac{1}{l_2} - \frac{1}{u_2}
    &= \frac{u_2 - l_2}{l_2 \cdot u_2}
    \\
    &\leq \epsilon_2 \cdot \frac{1}{l_2} \cdot \frac{1}{u_2}
    \\
    &\leq \epsilon_2 \cdot \left(\frac{1}{|E_2| - \epsilon_2}\right)^2
    && \text{because $E_2 \in [l_2, u_2]$}
    \\
    &\leq \epsilon_2 \cdot \left(\frac{1}{1/2^{B(E_2)} - \epsilon_2}\right)^2
    && \text{by Theorem~\ref{thm:sep}}
    \\
    &\leq \epsilon_2 \cdot \left(2 \cdot 2^{B(E_2)}\right)^2
    && \text{because $\frac{1}{2^{B(E_2)}} = \frac{4 \epsilon_2}{\epsilon'} \geq 2 \epsilon_2$}
    \\
    &\leq \epsilon_2 \cdot 4 \cdot 2^{B(E'_2)}
    && \text{because $B(E'_2) \geq 2 \cdot B(E_2)$}
    \\
    &= \epsilon'.
  \end{align*}
  %
  We can compute each reciprocal $1/u_2,1/l_2$ in $O(1)$ time by
  swapping the numerator and denominator,
  which preserves the number of bits $S(E_2,b_2)$.

  When we multiply $E'_2$ with $E_1$,
  the resulting number of bits $S(E,b)$ is the sum
  \begin{align*}
  & S(E_1, b') + S(E_2, b_2)
  \\
  ={} & S(E_1, b + B(E)) + S(E_2, b + B(E) + B(E'_2) + 2)
  \\
  \leq{} & S(E_1, b + B(E)) + S(E_2, b + 2 B(E) + 2)
  \\
  & \hspace{6em} \text{because $B(E'_2) = B(1/E_2) \leq B(E_1/E_2) = B(E)$}
  \\
  ={} &S(E, b).
  \end{align*}
  The total running time is $O(S(E,b))$ beyond the recursive calls, as claimed.

  \paragraph{Roots.}
  To $\epsilon$-compute $E = \sqrt[d]{E_1}$ for $d \geq 2$,
  we first recursively $\epsilon_1$-compute $E_1$
  where $\epsilon_1 = 4 (\epsilon/8)^d$.
  In other words, the recursive call is with $b_1 = d (b+3) + 2$.
  For cleaner formulas, we assume $b \geq 4$ (increasing $T(E,b)$ and
  $S(E,b)$ by only a constant factor, by Lemma~\ref{lem:T affine}) and
  round up $b_1$ to $2 d b$ (because $d \geq 2$).%
  \footnote{By assuming $b$ sufficiently large,
    the factor $2$ could be improved to any constant $>1$,
    which would improve the exponents in our final bounds
    by nearly a factor of~$2$.}
  Call the resulting interval $[l_1,u_1]$.
  If $l_1 \leq 0 \leq u_1$, then
  $$ |E_1| \leq \max\{|l_1|,|u_1|\} \leq
  u_1 - l_1 \leq \frac{1}{2^{b_1}} < \frac{1}{2^b} \leq
     \frac{1}{2^{B(E)}} \quad \text{(because $b \geq B(E)$)},
  $$
  so $E_1 = 0$ by Theorem~\ref{thm:sep}, and thus we can return~$[0,0]$.
  If $l_1 \leq u_1 < 0$, then $E_1 < 0$;
  if $d$ is even,
  then we return ``undefined'' to indicate a negative square root;
  if $d$ is odd,
  we can negate $l_1$ and $u_1$ to reduce to the positive case,
  and negate the final computed interval.

  In the remaining case, $0 < l_1 \leq u_1$.
  We can compute the floor and ceiling of the $d$th root
  of a $b$-bit integer in $O(d \cdot b)$ time
  \cite{Zimmermann-1999,Bertot-Magaud-Zimmermann-2002}; see \cite[\S 1.5.1]{Brent-Zimmermann-2010}.
  To compensate for the error to be introduced,
  we first scale the interval $[l_1,u_1]$ by $1/\epsilon_2^d$
  where $\epsilon_2 = \epsilon/8$,
  resulting in the interval
  $\left[
    l_1/\epsilon_2^d, u_1/\epsilon_2^d
  \right]$,
  which is an
  $\epsilon_1/\epsilon_2^d$-computation of $E_1/\epsilon_2^d$.
  %
  %
  We can simplify
  $\epsilon_1/\epsilon_2^d = \big( 4 (\epsilon/8)^d \big) / (\epsilon/8)^d = 4$,
  so we have a $4$-computation of $E_1/\epsilon_2^d$.
  Next we round to the containing integer interval
  $\left[
    \lfloor l_1/\epsilon_2^d \rfloor, \lceil u_1/\epsilon_2^d \rceil
  \right]$,
  which adds an additive error of at most $2$,
  so is
  a $6$-computation
  of $E_1/\epsilon_2^d$.
  Now we apply the $d$th-root algorithm to both ends of the integer interval,
  rounding down and up in each case respectively, to obtain
  $\left[
    \left\lfloor \sqrt[d]{\lfloor l_1/\epsilon_2^d \rfloor} \right\rfloor,
    \left\lceil \sqrt[d]{\lceil u_1/\epsilon_2^d \rceil} \right\rceil
  \right]$,
  which is an
  $8$-computation
  of $\sqrt[d]{E_1/\epsilon_2^d}
  = \sqrt[d]{E_1}/\epsilon_2$.
  Finally, we undo the scaling by multiplying both ends of this integer interval
  by $\epsilon_2$, and return
  $\left[
    \left\lfloor \sqrt[d]{\lfloor l_1/\epsilon_2^d \rfloor} \right\rfloor \epsilon_2,
    \left\lceil \sqrt[d]{\lceil u_1/\epsilon_2^d \rceil} \right\rceil \epsilon_2
  \right]$
  which is an $8 \epsilon_2 = \epsilon$%
  %
  %
  %
  %
  \hbox{-}computation of
  $\sqrt[d]{E_1} = E$.
  We can compute this interval using $O(1)$ rational multiplications
  via Lemma~\ref{lem:rationals} and the integer $d$th-root algorithm,
  in $O(S(E_1, b_1) d)$ time.
  The resulting number of bits $S(E,b)$ grows by
  $\lg \frac{1}{\epsilon_2^d} = d \lg \frac{8}{\epsilon} = d (b+3)$
  from the initial division by $\epsilon_2^d$;
  only decreases when we take the integer $d$th root;
  and grows by $\lg \frac{1}{\epsilon_2} = b+3$
  when we multiply by~$\epsilon_2$.
  Thus $S(E,b) = S(E_1, 2 d b) + (d+1)(b+3) \leq
  S(E_1, 2 d b) + 3 d b$ because $d \geq 2$ and $b \geq 4$.
\end{proof}

\begin{theorem} \label{thm:real/floor}
  Operations~\ref{op:real} and~\ref{op:floor} can be implemented
  on a word RAM in $O(T(E))$ and $O\big(T(E) \cdot S(E)\big)$ time respectively.
\end{theorem}

\begin{proof}
  Given an expression $E$, we apply Operation~\ref{op:realbits} to $E$
  with $b = 1 + B(E)$, which by Theorem~\ref{thm:realbits}
  can be performed in $O(T(E,b))$ time.
  By Corollary~\ref{cor:growth in b},
  $T(E,b) \leq 2 \cdot T(E,B(E)) = 2 \cdot T(E)$
  because $B(E) \geq 1$.
  Thus we obtain a rational interval $[l,u]$ of length $u - l \leq 2^b$
  that contains the value of~$E$.

  First we show how to compute the sign of $E$ in $O(1)$ additional time
  given this interval.
  If $0 \in [l,u]$ (i.e., $l \leq 0 \leq u$), then
  $|l|,|u| \leq 1/2^{1 + B(E)} = \frac{1}{2} 2^{B(E)}$,
  so by Theorem~\ref{thm:sep}, $E$ must in fact be zero.
  Otherwise, we have either $0 < l \leq u$, in which case $E$ must be positive;
  or $l \leq u < 0$, in which case $E$ must be negative.
  We can compute the signs of $l$ and $u$ in constant time
  by Lemma~\ref{lem:rationals}.

  Second we show how to compute the floor and ceiling,
  by reducing to a sign computation.
  Because $b \geq 1$, the interval length $u-l \leq \frac 1 2$,
  so $[l,u]$ contains at most one integer.
  We can compute $\lfloor l \rfloor$ and $\lceil u \rceil$
  using integer division of Lemma~\ref{lem:rationals}.
  By measuring the length of the expanded interval
  $[\lfloor l \rfloor, \lceil u \rceil]$,
  we determine whether $[l,u]$ contains an integer
  (the expanded interval has length $2$) or not
  (the expanded interval has length $1$).
  If $[l,u]$ does not contain an integer, i.e.,
  $i < l \leq u < i+1$ for an integer $i = \lfloor l \rfloor$,
  then $\lfloor E \rfloor = i$ and $\lceil E \rceil = i+1$.
  If $[l,u]$ contains an integer
  $i = \lfloor l \rfloor + 1 = \lceil u \rceil - 1$, then
  we compute the sign of $E - i$ using the algorithm above,
  which costs an additional $O(T(E-i))$ time.
  If $E - i$ is zero, then $\lfloor E \rfloor = \lceil E \rceil = i$;
  if $E - i$ is positive, then $\lfloor E \rfloor = i$ and
  $\lceil E \rceil = i+1$; and
  if $E - i$ is negative, then $\lfloor E \rfloor = i-1$ and
  $\lceil E \rceil = i$.
  Thus we obtain the floor and ceiling of $E$ in all cases.

  It remains to show that $T(E-i) = O\big(T(E) \cdot S(E)\big)$.
  By Lemma~\ref{lem:T minus}, $T(E-i) = O\big(T(E) \cdot T(i)\big)$.
  We have
  $T(i) = S(i) = C(i) = \max\{1,\lceil \lg |i| \rceil\} = O(S(E))$
  where the upper bound follows
  because $S(i)$ is the maximum number of bits
  in an integer in a rational approximation of $E \approx i$.
\end{proof}


\begin{corollary}
  Operations~\ref{op:combine},~\ref{op:real}, \ref{op:B}, \ref{op:realbits},
  and~\ref{op:floor} can be implemented on the word RAM
  in the time bounds given by the recursive cost model in Table~\ref{tab:ops}.
\end{corollary}

\begin{proof}
  By Lemma~\ref{lem:B}, Theorem~\ref{thm:realbits}, and
  Theorem~\ref{thm:real/floor}.
\end{proof}

\begin{corollary} \label{cor:compare}
  We can compare whether $E_1 \leq E_2$ for two expression DAGs
  $E_1$ and $E_2$ in $O\big(T(E_1) \cdot T(E_2)\big)$ time.
  %
\end{corollary}

\begin{proof}
  To compare $E_1$ and $E_2$, we apply
  Operation~\ref{op:real} from Theorem~\ref{thm:real/floor}
  to evaluate the sign of $E_1 - E_2$ in $O(T(E_1 - E_2))$ time.
  By Lemma~\ref{lem:T minus},
  $T(E_1 - E_2) = O\big(T(E_1) \cdot T(E_2)\big)$.
\end{proof}

\subsubsection{Simple Cost Model}
\label{sec:simple cost model}

Our second cost model is easier to use, but in general may be weaker.
It focuses on five key properties of an expression:
\begin{enumerate}
\item The \defn{height} $H(E)$ of the expression DAG $E$, i.e.,
  the length of the longest directed path,
  where integer expressions have height~$0$.
\item The \defn{number of nodes} $\#(E)$ of the expression DAG~$E$.
  In particular, $H(E)+1 \leq \#(E) < 2^{H(E)+1}$.
\item The \defn{root set} $R(E)$ of the expression DAG $E$,
  i.e., the set of distinct roots taken.
  (For example, $E = \sqrt{\sqrt 2 + \sqrt[3] 5} + \sqrt 2$ has
  $R(E) = \{\sqrt 2, \sqrt[3] 5, \sqrt{\sqrt 2 + \sqrt[3] 5}\}$.)
\item The \defn{maximum root degree}
  $K(E) = \max \{ k \mid \sqrt[k]{x} \in R(E) \}$.
  In our algorithm (and many computational geometry algorithms),
  $K(E) = 2$, meaning we just take square roots.
\item The maximum number $b^*(E)$ of \defn{bits} in the integer leaves of
  the expression DAG~$E$.
\end{enumerate}




%
%

The main challenge in obtaining a running time dependent only on the
number $|R(E)|$ of \emph{distinct} roots is in dealing with multiple
radical nodes of equal value, by detecting these equal-value roots and thereby
avoid paying the exponential cost for multiple occurrences of the same root.
First, in Lemma~\ref{lem:simplified bound} and
Theorem~\ref{thm:somewhat simple}, we ignore this issue,
and solve the $B(E)$ and $T(E)$ recurrences in terms of $D(E)$.
Next, in Lemma~\ref{lem:dedup radicals}, we show how to modify expressions
to remove equal-value radical nodes.
Finally, in Theorem~\ref{thm:simple}, we combine these tools to
derive the simple cost model.

\begin{lemma} \label{lem:simplified bound}
  $b^*(E) \leq B(E) \leq D(E) \cdot 2^{H(E)} \cdot b^*(E)$.
\end{lemma}

\begin{proof}
  It suffices to prove $b^*(E) \leq C(E) \leq 2^{H(E)} \cdot b^*(E)$
  Consider the recurrence \eqref{eq:C} for $C(E)$.
  The two recursive cases (second and third) have factors of $2$ that
  accumulate to $2^d$ for a node of depth $d$,
  where $d \leq H(E)$.
  The base case (first) is at most $b^*(E)$,
  and in at least one case is exactly $b^*(E)$.
  Thus $b^*(E) \leq C(E) \leq 2^{H(E)} \cdot b^*(E)$.
\end{proof}

\begin{corollary} \label{cor:simple B}
  The running time $O(\#(E) \lg B(E))$ for Operation~\ref{op:B}
  from Lemma~\ref{lem:B} is bounded by
  $$O\Big( \#(E)^2 \lg K(E) + \#(E) \lg b^*(E) \Big).$$
\end{corollary}

\begin{proof}
  By Lemma~\ref{lem:simplified bound},
  $\lg B(E) \leq \lg D(E) + H(E) + \lg b^*(E)$.
  By definition, $\lg D(E) \leq \#(E) \lg K(E)$
  and $H(E) \leq \#(E)$.
\end{proof}

\begin{theorem} \label{thm:somewhat simple}
  The running time $T(E,b)$ for Operation~\ref{op:realbits} satisfies
  $$T(E,b) = O\left(b \cdot \big(8 \cdot K(E)\big)^{H(E)+2}\right).$$
  The running time $T(E)$ for Operation~\ref{op:real} satisfies
  $$T(E) = O\left(D(E) \cdot b^*(E) \cdot \big(16 \cdot K(E)\big)^{H(E)+2}\right).$$
\end{theorem}

\begin{proof}
  We analyze the growth of the $b$ parameter
  in the $S(E,b)$ and $T(E,b)$ recurrences of \eqref{eq:S} and~\eqref{eq:T},
  which follow the same pattern in terms of how $b$ grows down the recursion.
  Let $b_0$ represent the parameter $b$ at the top level of the recurrence,
  i.e., in the call to Operation~\ref{op:realbits}.
  By Lemma~\ref{lem:T affine}, we can assume that $b_0 \geq B(E) + 1$
  (instead of just $b_0 \geq B(E)$) by increasing our bound by a constant factor.
  Each radical node in $E$ recurses with $b' = 2 d b$,
  thus multiplying the input $b$ by a factor of $2 d \leq 2 \cdot K(E)$.
  Each nonradical node in $E$ recurses with $b' \leq b + 2 B(E) + 2
  \leq 3 b$ because $b \geq b_0 \geq B(E) + 1$, thus multiplying the input $b$
  by a factor of at most $2 \leq 2 \cdot K(E)$.
  Thus each recursive level increases $b$
  by a factor of at most $2 \cdot K(E)$.
  Then after $H(E)$ levels, we multiply $b$ by a factor of at most
  $(2 \cdot K(E))^{H(E)}$.
  Thus we can assume throughout the $S(E,b)$ and $T(E,b)$ recursions
  that $b \leq b_0 \cdot (2 \cdot K(E))^{H(E)}$.

  Now we compute $S(E,b)$ via the recurrence \eqref{eq:S}.
  The number of recursive calls to $S(E,b)$ is at most $2^{H(E)}$.
  Each recursive call has an additive term of
  $b' \leq b^*(E) \leq B(E) \leq b_0$ in the base case
  (by Lemma~\ref{lem:simplified bound}),
  and at most $3 d b \leq 4 \cdot K(E) \cdot b_0 \cdot (2 \cdot K(E))^{H(E)}
  = O(b_0 \cdot (2 \cdot K(E))^{H(E)+1})$
  in the recursive case (dominated by the radical case).
  Multiplying, the total cost is at most
  $2^{H(E)} \cdot O(b_0 \cdot (2 \cdot K(E))^{H(E)+1})
  = O(b_0 \cdot (4 \cdot K(E))^{H(E)+1})$.


  Finally we compute $T(E,b)$ via the recurrence \eqref{eq:T}.
  The number of recursive calls to $T(E,b)$ is at most $2^{H(E)}$.
  Each recursive call has an additive term of at most
  $d \cdot S(E,b) \leq K(E) \cdot O(b_0 \cdot (4 \cdot K(E))^{H(E)+1})$.
  Multiplying, the total cost is at most
  $2^{H(E)} \cdot K(E) \cdot O(b_0 \cdot (4 \cdot K(E))^{H(E)+1})
  = O(b_0 \cdot (8 \cdot K(E))^{H(E)+2})$.
  We obtain the desired upper bound on $T(E) = T(E,B(E))$
  by substituting Lemma~\ref{lem:simplified bound}'s upper bound
  $B(E) \leq D(E) \cdot 2^{H(E)} \cdot b^*(E)$.
  %
\end{proof}

Define a \defn{uniquification} of expression $E$
to be an expression $E'$
satisfying $E' = E$, $\#(E') \leq \#(E)$, $H(E') \leq H(E)$,
$K(E') = K(E)$, $b^*(E') = b^*(E)$, $R'(E) = R(E)$, and
the number of radical nodes in $E'$ is
the number $|R(E)| = |R(E')|$ of distinct roots in~$E$
(i.e., all roots in $E'$ have distinct values).
Define a \defn{near-uniquification} in the same way, but allowing
the number of radical nodes in $E'$ to be at most $1$ larger
(i.e., $\leq |R(E)| + 1 = |R(E')| + 1$).

\begin{lemma} \label{lem:dedup radicals}
  Given an expression DAG $E$, we can compute a uniquification $E'$ of $E$ in
  $$O\left(\left(\sum_{i=1}^{\#(E)} T(E_i) \right)^2 ~ \right)$$
  time, where each $E_i$ is a near-uniquification of a subexpression of~$E$.
\end{lemma}

\begin{proof}
  First we define a partial order on pairs of radical nodes of~$E$.
  Let $\mathcal R$ be the set of radical nodes in~$E$.
  Define a partial order $\preceq$ on $\mathcal R$ by $x \preceq y$ when
  $x$ is descendant of, or the same node as, $y$ in~$E$.
  This partial order induces a partial order on the cross product
  $\mathcal R \times \mathcal R$ by $(x,x') \preceq (y,y')$ when
  $x \preceq y$ and $x' \preceq y'$.

  Our algorithm compares all pairs $(x,x')$ of distinct radical nodes
  in $\mathcal R$ in a linearization of~$\preceq$ (found via topological sort),
  removing any duplicate (equal-value) nodes by combining those nodes.
  Thus, when we visit a pair $(x,x')$, we know that the
  expression DAG $E_x$ with source $x$ and the expression DAG $E_{x'}$
  with source $x'$ have already had all of their radical nodes pairwise compared
  and deduplicated, except for the single pair $(x,x')$.
  In particular, $E_x$ and $E_{x'}$ are near-uniquifications of the
  original expression DAGs with source nodes $x$ and $x'$ respectively.
  We apply Corollary~\ref{cor:compare} to compare whether $E_x = E_{x'}$.
  If so, we have detected a duplicate root, and we remove it as follows.
  If $x'$ is a descendant of $x$, then exchange $x$ and~$x'$,
  so that $x$ is a descendant of~$x'$.
  If $x$ and $x'$ are incomparable by the descendant relation, and $H(x) > H(x')$, then exchange $x$ and~$x'$.
  Then, in all cases, $x'$ is not a descendant of $x$
  and $H(x) \leq H(x')$.
  Now replace in $E$ all references to $x'$ by references to $x$
  (which preserves acyclicity because $x'$ is not a descendant of~$x$,
  and preserves or decreases height because $H(x) \leq H(x')$),
  and then remove $x'$ from the DAG~$E$.
  As a special case, if $x'$ is the source node of the expression DAG~$E$,
  then we replace the entire DAG with~$E_x$.
  The new $E_x$ is a uniquification of a subexpression of~$E$
  (the original expression DAG with source~$x$).
  Repeating this process, we remove all duplicate roots from $E$,
  while preserving the values of $E$, $K(E)$, $b^*(E)$, and $R(E)$,
  and only decreasing $\#(E)$ and $H(E)$,
  resulting in a uniquification of~$E$.

  Finally we analyze the running time of this algorithm.
  For each of ${|\mathcal R| \choose 2} \leq \#(E)^2$
  pairs of distinct radical nodes $(x,x')$, we pay
  $O\big(T(E_x) \cdot T(E_{x'})\big)$ time by Corollary~\ref{cor:compare},
  where $E_x$ and $E_{x'}$ are near-uniquifications of subexpressions of~$E$.
  Therefore the total time bound is
  \begin{align*}
  O\left(\sum_{i=1}^{\#(E)} \sum_{j=1}^{\#(E)} T(E_i) \cdot T(E_j)\right)
  = O\left(\sum_{i=1}^{\#(E)} T(E_i) \cdot \sum_{j=1}^{\#(E)} T(E_j)\right)
  = O\left(\left(\sum_{i=1}^{\#(E)} T(E_i) \right)^2 ~ \right).
  \end{align*}
\end{proof}

\begin{theorem} \label{thm:simple}
  Operations~\ref{op:real} and~\ref{op:floor}
  can be implemented in
  \begin{align*}
  & O\left( \#(E)^2 \cdot
      b^*(E)^2 \cdot \big(16 \cdot K(E)\big)^{2 |R(E)| + 2 H(E)+6}
    \right)
  \\
  ={}&
    O\left(
      b^*(E)^2 \cdot \big(32 \cdot K(E)\big)^{2 |R(E)| + 2 H(E)+6}
    \right)
  \end{align*}
  time.
\end{theorem}

\begin{proof}
  We implement each operation by first removing all duplicate roots
  from $E$ via Lemma~\ref{lem:dedup radicals},
  resulting in a uniquification $E'$ without duplicate roots,
  and second by running the appropriate algorithm from
  Theorem~\ref{thm:real/floor} on input~$E'$.

  By Lemma~\ref{lem:dedup radicals},
  the running time for the first step is given by
  $$O\left(\left(\sum_{i=1}^{\#(E)} T(E_i) \right)^2 ~ \right).$$
  By Theorem~\ref{thm:somewhat simple}, each $T(E_i)$ satisfies
  $$T(E_i) = O\left(D(E_i) \cdot b^*(E_i) \cdot \big(16 \cdot K(E_i)\big)^{H(E_i)+2}\right).$$
  Because each $E_i$ is a near-uniquification,
  it has at most $|R(E_i)|+1$ radical nodes,
  so $D(E_i) \leq K(E_i)^{|R(E_i)|+1}$.
  And because each $E_i$ is a near-uniquification of a subexpression of~$E$,
  we have $b^*(E_i) \leq b^*(E)$, $K(E_i) \leq K(E)$,
  $|R(E_i)| \leq |R(E)|$, and $H(E_i) \leq H(E)$.
  Thus
  \begin{align*}
    T(E_i)
    &= O\left(K(E)^{|R(E)|+1} \cdot b^*(E) \cdot \big(16 \cdot K(E)\big)^{H(E)+2}\right)
    \\
    &= O\left(b^*(E) \cdot \big(16 \cdot K(E)\big)^{|R(E)|+H(E)+3}\right)
    ,
  \end{align*}
  so the running time of the first step is
  \begin{align*}
  & O\left(\left(\sum_{i=1}^{\#(E)}
    b^*(E) \cdot \big(16 \cdot K(E)\big)^{|R(E)|+H(E)+3}
  \right)^2 ~ \right)
  \\
  &= O\left( \#(E)^2 \cdot
      b^*(E)^2 \cdot \big(16 \cdot K(E)\big)^{2 |R(E)| + 2 H(E)+6}
    \right).
  \end{align*}
  Because $\#(E) < 2^{H(E)+1}$,
  we have $\#(E)^2 < 2^{2 \cdot H(E)+2}$,
  whose exponent is smaller than $2 |R(E)| + 2 H(E)+6$,
  so we can absorb $\#(E)^2$ by increasing the base by a factor of~$2$.

  By Theorem~\ref{thm:real/floor} on input~$E'$,
  the running time for the second step is
  $O(T(E'))$ for Operation~\ref{op:real} and
  $O\big(T(E') \cdot S(E')\big) = O(T(E')^2)$ for Operation~\ref{op:floor}
  (where $S(E') \leq T(E')$ by inspection of \eqref{eq:T}).
  Because $E'$ is a uniquification of~$E$,
  $|R(E')|$ is the number of radical nodes in~$E'$,
  so $D(E') \leq K(E')^{|R(E')|} = K(E)^{|R(E)|}$.
  By Theorem~\ref{thm:somewhat simple},
  \begin{align*}
  T(E')
  &= O\left(K(E)^{|R(E)|} \cdot b^*(E) \cdot \big(16 \cdot K(E)\big)^{H(E)+2}\right).
  \end{align*}
  Therefore the running time is dominated by the first step.
\end{proof}

The same idea can be used to ``speed up'' Operation~\ref{op:realbits}
to require only $b \geq B(E')$ instead of $b \geq B(E)$.
Then the running time would be the sum of the costs from
Theorem~\ref{thm:simple} (for the first step) and
Theorem~\ref{thm:somewhat simple} (for the second step).

\subsubsection{Special Cases}

Now we look at some restricted expressions where the simple cost model
is particularly simple.
Our first restriction is to $O(1)$-size expressions,
where real operations run as fast as integer operations.
(An earlier version of this paper \cite{Quasigeodesics_SoCG2020}
called this model the \defn{$O(1)$-expression RAM}, and we suspect it is
useful in many computational geometry algorithms; sadly, it will not suffice
for our quasigeodesics algorithm.)

\begin{corollary}
  If $\#(E) = O(1)$ and $K(E) = O(1)$, then $T(E) = O(b^*(E))$.
\end{corollary}

\begin{proof}
  Follows from Theorem~\ref{thm:somewhat simple},
  $H(E) \leq \#(E)$, and
  $D(E) \leq K(E)^{2^{H(E)}} = O(1)$.
\end{proof}

Next we consider expressions of logarithmic height, where the running time
for real operations is polynomial except for an exponential dependence on
the number of distinct roots.
This result is what we will use in our quasigeodesics algorithm.

\begin{corollary}
  If $H(E) = O(\lg n)$, $b^*(E) = n^{O(1)}$, and $K(E) = O(1)$, then
  Operations~\ref{op:real} and~\ref{op:floor} can be implemented in
  $n^{O(1)} \cdot 2^{O(|R(E)|)}$ time.
\end{corollary}

\begin{proof}
  Follows from Theorem~\ref{thm:simple} and $\#(E) \leq 2^{H(E)}$.
\end{proof}

Finally, we prove a general exponential bound on the running time
for real operations:

\begin{corollary}
  If $\#(E) = O(n)$ nodes, $b^*(E) = 2^{O(n)}$, and $K(E) = O(1)$, then
  Operations~\ref{op:real} and~\ref{op:floor} can be implemented in
  $2^{O(n)}$ time.
\end{corollary}

\begin{proof}
  Follows from Theorem~\ref{thm:simple},
  $H(E) \leq \#(E)$, and $|R(E)| \leq \#(E)$.
\end{proof}

\begin{corollary}
  Any algorithm running in $2^{O(n)}$ time on the real RAM,
  where all radical operations are of $O(1)$ degree,
  can be run in $2^{O(n)}$ time on the word RAM.
\end{corollary}


\subsubsection{Application to Multi-Expression Objects}
\label{sec:multi-expression}

We provide some notation to more easily express the simple cost model
for algorithms whose input is not one expression but an ``object'' $M$
represented by multiple expressions.
Let $\mathcal E(M)$ be the set of expressions representing~$M$.
First we define parameters $H(M)$, $b^*(M)$, $K(M)$ and $R(M)$ as follows:
\begin{align*}
  H(M) &= \max_{E \in \mathcal E(M)} H(E), &
  b^*(M) &= \max_{E \in \mathcal E(M)} b^*(E), &
  K(M) &= \max_{E \in \mathcal E(M)} K(E), &
  R(M) &= \bigcup_{E \in \mathcal E(M)} R(E).
\end{align*}

Now we define $T'(M)$ to be the simple-cost-model running time
of Theorem~\ref{thm:simple} according to these parameters:
\begin{align} \label{eq:T'(M)}
T'(M) &=
    \Theta\left(
      b^*(M)^2 \cdot \big(32 \cdot \max\{2,K(M)\}\big)^{2 |R(M)| + 2 H(M) + 6}
    \right),
\end{align}

Any algorithm will construct expressions in terms of the input expressions
in $\mathcal E(M)$.
We define $T'(M,x)$ to be the simple-cost-model running time
of Theorem~\ref{thm:simple} for a height-$h$ expression whose leaves
are in fact input expressions in $\mathcal E(M)$,
with at most $r$ additional square roots beyond the roots in $R(M)$,
such that $h+r \leq x$ (so $x$ represents the ``extra complexity''):
\begin{align} \label{eq:T'(M,h)}
T'(M,x) &=
    \Theta\left(
      b^*(M)^2 \cdot \big(32 \cdot \max\{2,K(M)\}\big)^{2 |R(M)| + 2 H(M) + 2 x + 6}
    \right),
\\
T'(M) &= T(M,0). \nonumber
\end{align}

We treat a tuple $(M_1, M_2, \ldots, M_k)$ of objects
as itself an object,
whose expression set $\mathcal E$ is $\bigcup_{i=1}^k \mathcal E(M_i)$.

\subsection{Associative Operations}

\begin{lemma} \label{lem:associative}
  Consider an associative operator $\circ$ over $d$-dimensional real vectors,
  defined by $d$ expressions (forming $\mathcal E(\circ)$)
  whose leaves can also include one of the $2 \cdot d$ input variables.
  Given $d$-dimensional vectors $x_1, x_2, \dots, x_n$,
  we can compute in $O(n)$ time a representation of
  $x_1 \circ x_2 \circ \cdots \circ x_n$ as a $d$-dimensional vector $F$
  satisfying
  \begin{align*}
    H(F) &= \lceil \lg n \rceil \cdot H(\circ) + \max_{k=1}^n H(x_k), &
    b^*(F) &= \max\left\{b^*(\circ), \max_{k=1}^n b^*(x_k)\right\}, \\
    R(F) &= R(\circ) \cup \bigcup_{k=1}^n R(x_k), &
    K(F) &= \max\left\{K(\circ), \max_{k=1}^n K(x_k)\right\}.
  \end{align*}
  In particular, $T'(x_1 \circ x_2 \circ \cdots \circ x_n) \leq
  T'\big((x_1, x_2, \dots, x_n), \lceil \lg n \rceil \cdot H(\circ)\big)$.
\end{lemma}

\begin{proof}
  At a high level, we compute $x_1 \circ x_2 \circ \cdots \circ x_n$
  according to a complete rooted ordered binary tree $T$,
  with one leaf for each $x_i$ in order and
  where each internal node computes~$\circ$.
  This tree has height $\lceil \lg n \rceil$.
  To transform this vector expression into real expressions,
  we make $d$ copies $T_1, T_2, \dots, T_d$ of $T$.
  In $T_i$, we replace the $k$th leaf with
  the expression representing the $i$th element of~$x_k$,
  and we replace each internal node $v$ with
  the expression representing the $i$th element of~$\circ$,
  where each reference to the $j$th element of the left [right] expression
  becomes a pointer to the left [right] child of $v$ in~$T_j$.
  Then $F = (s_1, s_2, \dots, s_d)$
  where $s_i$ is the source node of expression DAG~$T_i$.
  The bound on $T$ follows from Equation~\eqref{eq:T'(M,h)}.
\end{proof}

\subsection{Polyhedral Inputs}
\label{Polyhedral Inputs}

The combinatorial structure of an input polyhedron can be encoded as a primal
or dual graph, as usual, but which real numbers should represent the geometry?
Because the quasigeodesic problem is about the intrinsic geometry of the
surface of a polyhedron, the input geometry for this problem is most
naturally represented intrinsically.
More generally, we have three natural encodings of a polyhedron:

\begin{enumerate}
\item \label{model:intrinsic}
  \textbf{Intrinsic coordinates:} For each face,
  for some isometric embedding of the face into 2D,
  the 2D coordinates of each vertex of the embedded face.
  This representation is the primary model assumed by our algorithm.
  In fact, it can be used to represent any \defn{Alexandrov gluing}
  \cite{Demaine-O'Rourke-2007}, where the gluing forms a topological sphere
  and at most $360^\circ$ of face angle gets glued at each vertex.
  In this situation, the encoded ``faces'' may not be extrinsically
  planar on the convex surface, but they would be intrinsically planar
  on the unique convex polyhedron guaranteed by Alexandrov's Theorem
  \cite{Alexandrov-1996,Demaine-O'Rourke-2007}.
\item \label{model:length}
  \textbf{Intrinsic lengths:} For each face, the lengths of the edges.
  This representation assumes the faces have been combinatorially
  triangulated (so some edges may be flat).
\item \label{model:extrinsic}
  \textbf{Extrinsic coordinates:} 3D coordinates for each vertex.
\end{enumerate}

In the expression RAM, we can convert coordinates (\ref{model:intrinsic} or \ref{model:extrinsic}) to
edge lengths (\ref{model:length}) as follows:
given vertex coordinates $(x_1,y_1,z_1), \allowbreak (x_2,y_2,z_2)$,
the distance between these two vertices is given by the constant-size
radical expression:
$$\sqrt{(x_1 - x_2) \cdot (x_1 - x_2) +
        (y_1 - y_2) \cdot (y_1 - y_2) +
        (z_1 - z_2) \cdot (z_1 - z_2)}.$$
The added roots from this transformation is exactly
the set of edge lengths $\Lambda$ in a triangulation of the polyhedron
(as needed by representation (\ref{model:length})).

We can also convert from intrinsic lengths (\ref{model:length}) to intrinsic coordinates (\ref{model:intrinsic})
by, for each triangle, placing two vertices on the $x$ axis and
finding the third vertex by the intersection of two circles
whose radii are the two incident edge lengths.
This transformation requires a square root,
potentially one for each distinct triangle in the input.
Thus the number of distinct added roots from this transformation is
at most $|\Lambda|^3$,
because we can form only $|\Lambda|^3$ distinct triangles
from $|\Lambda|$ edge lengths.

Therefore we can convert representations
$$
\text{(\ref{model:extrinsic})} \underset{|\Lambda|}{\rightarrow} \text{(\ref{model:length})} \overset{\strut|\Lambda|}{\underset{|\Lambda|^3}{\leftrightarrows}} \text{(\ref{model:intrinsic})},
$$
where the labels specify the number of added roots.
(The reverse direction, from intrinsic (\ref{model:intrinsic}/\ref{model:length}) to extrinsic (\ref{model:extrinsic}),
is more difficult, as it involves solving the Alexandrov problem
\cite{Alexandrov_WADS2009}.
Accordingly, the intrinsic representations (\ref{model:intrinsic}/\ref{model:length}) represent
a more general class of possible polyhedra.
By contrast, in a model restricting inputs to be integers or rationals,
these three input models would define incomparable classes of polyhedra,
so no representation conversions would be possible.)

Our quasigeodesic algorithm assumes the intrinsic input representation (\ref{model:intrinsic}).
By the reductions above, our algorithm also applies to polyhedra
given in the extrinsic representation (\ref{model:extrinsic})
or intrinsic length representation (\ref{model:length}).
On a real RAM, these conversions incur only linear additive time cost.
On a word RAM, these conversions add up to $|\Lambda|^3$ roots
to the root set $R$ of the input expressions.
Thus the $2^{O(|R|)}$ term in the running time increases to
$2^{O(|R| + |\Lambda|^3)}$ for these input representations.

\section{Expression RAM Algorithm}
\label{ExpressionRAMAlgorithmSection}

In this section, we describe how to re-implement
Section~\ref{RealRAMAlgorithmSection}'s real RAM algorithm
for finding closed quasigeodesics onto the expression RAM
of Section~\ref{sec:Expression RAM},
and analyze the resulting running time on a word RAM.

In the sense of Section~\ref{sec:multi-expression},
we define a \defn{polyhedron object} $P$ to be represented by expressions
$\mathcal{E}(P)$ consisting of,
for each face $f$ of $P$,
the $x$ and $y$ coordinates for each vertex of $f$ in an isometric
2D embedding of~$f$ (as in representation~(\ref{model:intrinsic})).

\subsection{Computing Quasigeodesic Rays}

\setcounter{theorem}{4}

\begin{lemma}[Expression RAM version of Lemma~\ref{lem:coordinate system transforms}]
  \label{lem2:coordinate system transforms}
  Given the intrinsic coordinates $C_1,C_2$ for two adjacent faces $f_1,f_2$
  sharing an edge~$e$, where each $C_i$ is a vector $C$ specifying
  the coordinates of all vertices of~$f_i$, we can compute
  the orientation-preserving isometry
  bringing $e$ on $f_1$ to $e$ on $f_2$ as a transformation matrix~$I$,
  where $H(I) \leq H\big((C_1,C_2)\big) + O(1)$ and
  $b^*$, $R$, and $K$ are the same for $I$ as for $(C_1,C_2)$.
  %
  On the word RAM, the running time is $O(1)$.
\end{lemma}

\begin{proof}
  As stated in the proof of Lemma~\ref{lem:coordinate system transforms},
  the affine transformation $I$ can be written as a $3 \times 3$
  matrix using $O(1)$ additions, subtractions, multiplications, and divisions
  of the input coordinates.
  Thus $H$ increases by at most an additive constant.
  There are no new leaf expressions or roots.
  Because we only use Operation~\ref{op:combine},
  the running time is $O(1)$.
\end{proof}

\begin{lemma}[Expression RAM version of Lemma~\ref{lem:ray follow exact}]
  \label{lem2:ray follow exact}
  Let $S = (X, \vec v, \infty)$ be a geodesic ray
  on an $n$-vertex polyhedron,
  where point $X$ is on the boundary of a face $f$,
  $\vec v$ points inside $f$ from $X$, and
  $X$ and $\vec v$ are specified in the local coordinate system of~$f$.
  We can compute the first $k$ faces and $k$ edges visited by~$S$;
  the corresponding $k$ intersection points; and
  a planar embedding of an unfolding of these faces and edges.
  In particular, this determines the direction at which $S$ exits the
  last face into the last intersection point.
  All points and directions are represented as (exact) expressions
  of height $H\big((P,S)\big) + 3 \lg k + O(1)$ and root set
  $R\big((P,S)\big)$.%
  %
  \footnote{For the purposes of $H(S)$ and $R(S)$,
    we ignore the $\infty$ component, and just view $S$ as the vector object
    $(X, \vec v)$; see Section~\ref{sec:multi-expression}.}
  If $S$ hits a vertex within the first $k$ steps, we stop there and
  also output the vertex; otherwise, we proceed for exactly $k$ steps.
  On the word RAM, the running time is
  $O\Big(T'\big((P,S), 3 \lg k + O(1)\big) \, k \lg n\Big)$.
\end{lemma}

\begin{proof}
  In the algorithm in the proof of Lemma~\ref{lem:ray follow exact},
  we can implement a single step $i$ as follows.

  During each step of the binary search,
  we compute $Q'$ by solving the $O(1)$-size linear system
  using $O(1)$ arithmetic operations (via \ref{op:combine})
  in terms of $A,B$ (part of~$P$), $\vec v_{i-1}$, and~$Q_{i-1}$.
  Then we use Type-\ref{op:real} comparisons to
  to compare the signed triangle areas with~$0$.
  According the simple cost model, this operation takes
  $O\Big(T'\big((P, Q_{i-1}, \vec v_{i-1})\big)\Big)$ time
  per step of the binary search.
  Thus the binary search finds the exit edge $e_i$ in
  $O\Big(T'\big((P, Q_{i-1}, \vec v_{i-1})\big) \lg n\Big)$ time.

  To reduce the expression complexity of $Q_i$ and $\vec v_i$,
  at each step $i$,
  we apply Lemma~\ref{lem:associative} to represent each $(Q_i, \vec v_i)$
  as the associative composition (matrix multiplication)
  of $\tau_1, \dots, \tau_{i-1}$ applied in that order to $(Q, \vec v)$.
  Each $\tau_i$ is the composition/product of $\rho_i$ and $\sigma_i$, where
  $\rho_i$ transforms from one edge of $f_i$ to another edge $e_i$ of $f_i$,
  and $\sigma_i$ transforms from the latter edge $e_i$ of $f_i$ to
  the corresponding edge of $f_{i+1}$
  (given by Lemma~\ref{lem2:coordinate system transforms}).
  Both $\rho_i$ and $\sigma_i$ (and thus $\tau_i$) are given by matrices
  whose elements are $O(1)$-expressions in terms of
  the expressions $\mathcal E(P)$ representing polyhedron~$P$,
  with no new roots.
  Each cell of the $3 \times 3$ matrix multiplication can be written
  as a height-3 expression (a height-2 sum of three products).
  By Lemma~\ref{lem:associative},
  this representation of $(Q_i,\vec v_i)$ uses expressions of
  height $3 \lceil \lg (2i) \rceil + O(1)
  = 3 \lg i + O(1)$
  in terms of the expressions $\mathcal E(P)$ representing polyhedron~$P$.
  Thus $T'\big((P, Q_i, \vec v_i)\big)
  = O\Big(T'\big((P,S), 3 \lg i + O(1)\big)\Big)$.
  We also pay $O(i)$ time to apply Lemmas~\ref{lem:associative}
  and~\ref{lem2:coordinate system transforms},
  but by Equation~\eqref{eq:T'(M,h)},
  this cost is dominated by
  $$T'\big((P,S), 3 \lg i + O(1)\big) = \Omega\big(2^{3 \lg i + O(1)}\big) = \Theta(i^3).$$

  Therefore, the total time for step $i$ is
  $O\Big(T'\big((P,S), 3 \lg i + O(1)\big) \lg n\Big)$.
  Summing over $i$,
  we obtain the claimed time bound
  $O\Big(T'\big((P,S), 3 \lg k + O(1)\big) \, k \lg n\Big)$.

  Finally, to compute the planar embedding of an unfolding of the
  faces and edges
  in the local coordinate system of~$f_1$,
  we use Lemma~\ref{lem:associative} to represent the associative composition
  of $\sigma^{-1}_{i-1}, \sigma^{-1}_{i-2}, \dots, \sigma^{-1}_1$
  applied in that order, and then apply that transformation to
  $f_i$ and $e_i$.
  Because $f_i$ and $e_i$ are part of
  the expressions $\mathcal E(P)$ representing polyhedron~$P$,
  we add $3 \lceil \lg i \rceil = \lg k + O(1)$ height
  and no new roots.
\end{proof}

\begin{corollary}[Expression RAM version of Corollary~\ref{cor:cone following}]
  \label{cor2:cone following}
  Consider an angle-$\theta$ cone between geodesic rays
  $S_1 = (X,\vec v_1, \infty)$ and $S_2 = (X,\vec v_2, \infty)$.
  We can compute a geodesic segment $S = (X, \vec v, d)$
  that is in the given cone and ends at a vertex $Y$ of~$P$
  and has $k = O\left(\frac{L^2}{\theta \, \ell^2}\right)$
  intersections between $S$ and edges of~$P$,
  along with the identity of the faces and edges intersected by~$S$.
  The output consists of $\vec v$
  and the $k$ intersections,
  represented by (exact) expressions of height
  $H\big((P,S_1,S_2)\big) + 6 \lg k + O(1)$ and
  root set $R\big((P,S_1,S_2)\big)$.
  On the word RAM, the running time is
  $O\Big(T'\big((P,S_1,S_2), 6 \lg k + O(1)\big) \, k \lg n\Big)$.
\end{corollary}

\begin{proof}
  We follow the algorithm of Corollary~\ref{cor:cone following},
  but replacing Lemma~\ref{lem:ray follow exact} with
  Lemma~\ref{lem2:ray follow exact}.
  This lemma gives us the identity of the faces and edges intersected by
  $S_1$ and $S_2$ and thus~$S$, as well as an unfolding of these faces and edges
  in the local coordinate system of~$f_1$.
  We can compute $\vec v$ as the difference between points $Y$ and $X$
  in this unfolding, which adds only constant height to the expression.
  We can compute each intersection $X_i$ in the local coordinate system of
  $f_1$ by intersecting the corresponding edge $e_i$ in the unfolding
  with the segment $X Y$ in the local coordinate system of~$f_1$,
  which again adds constant height to these expressions.
  We can map this intersection point $X_i$ to the local coordinate system of
  $f_i$ (which better reflects a ``point on the polyhedron $P$'')
  by applying the transforms $\sigma_1, \sigma_2, \dots, \sigma_{i-1}$
  in that order, using Lemma~\ref{lem:associative}, which adds an
  additional $3 \lg k + O(1)$ to the height of the expressions
  (as in Lemma~\ref{lem2:ray follow exact}).
  Thus the total height is $6 \lg k + O(1)$
  beyond the height of the input expressions
  in $\mathcal E(P)$, and the root sets remain unchanged.%
  \footnote{Likely this constant $6$ could be improved. For example, we may be
    able to combine the two matrix products into one tree, reducing $6$ to $3$.}
\end{proof}

\subsection{Full Algorithm}

\begin{theorem}[Expression RAM version of Theorem~\ref{FindQuasigeodesicTheorem}]
\label{thm2:FindQuasigeodesicTheorem}
Let $P$ be a convex polyhedron with $n$ vertices all of curvature at least
$\epsilon$ and degree at most $\Delta$,
let $L$ be the length of the longest edge, and let $\ell$ be the
smallest distance within a face between a vertex and a nonincident edge.
Then we can find a closed quasigeodesic on $P$ consisting of
$O\left(\frac{n \, L^2}{\epsilon^2 \, \ell^2}\right)$
segments on faces of~$P$.
On an expression RAM, each segment endpoint of the closed quasigeodesic
can be represented by an expression with root set $R(P)$
plus one additional root,
and height $H(P) +
2 \lg \tfrac{1}{\epsilon} + 3 \lg \Delta + 6 \lg \tfrac{L^2}{\epsilon \, \ell^2} + O(1)$,
and the running time is
$$
O\left( b^*(P)^2 \cdot \big(32 \cdot \max\{2,K(P)\}\big)^{2 |R(P)| + 2 H(P) + 16 \lg \tfrac{1}{\epsilon} + 6 \lg \Delta + 24 \lg \tfrac{L}{\ell} + O(1)} \cdot \tfrac{L^2}{\epsilon^2 \, \ell^2} \, n \lg n \right).
$$
In particular, if $K(P) \leq 2$, this running time is
$$O\left(b^*(P)^2 \cdot 2^{12 \big( |R(P)| + H(P)\big)}
\cdot
\tfrac{\Delta^{36} \, L^{146}}{\epsilon^{98} \, \ell^{146}} \, n \lg n
\right).$$
\end{theorem}

\begin{proof}
We follow the proof of Theorem~\ref{FindQuasigeodesicTheorem}.

First, we compute the unit vector $\vec v_\epsilon$
whose counterclockwise angle from the positive $x$~axis
is the minimum curvature $\epsilon$ of the polyhedron's vertices.
At a vertex $V$ with incident faces $f_1, f_2, \dots, f_k$,
we apply Lemma~\ref{lem2:coordinate system transforms}
to compute transforms from the local coordinate system of each face $f_i$
to the local coordinate system of the previous face $f_{i-1}$,
or for $i=1$, to place the vertex $V$ at the origin with the appropriate edge
of $f_1$ on the positive $x$~axis with $f_1$ below the axis.
By Lemma~\ref{lem2:coordinate system transforms},
these transformations add only constant height and no new roots.
To compute each placed face $f_i$,
we apply Lemma~\ref{lem:associative} to compute the
(associative) product of the transformation matrices from
Lemma~\ref{lem2:coordinate system transforms}, and then
apply that transformation to compute the vertex coordinates of the placed~$f_i$.
As in Lemma~\ref{lem2:ray follow exact}, this adds
$3 \lg k + O(1)$ to the height.
An edge $e$ of the last placed face $f_k$, divided by its length $\|e\|$
(which involves a square root),
gives a candidate vector for $\vec v_\epsilon$,
with added height at most $3 \lg \Delta + O(1)$
and one additional square root.
To determine which unit vector is $\vec v_\epsilon$, i.e., closest to the $x$ axis,
we compute the slope of each vector (the ratio of the two coordinates)
and the signs of the two coordinates (determining the quadrant).
Then we can use Operation~\ref{op:real} to compare two vectors to find
which has the smaller angle in the better quadrant, and via a linear scan,
find the vector $\vec v_\epsilon$ with the smallest angle in the best quadrant.
The total running time is
\begin{equation} \label{eq:T1}
O\big(n \cdot T'(P, 3 \lg \Delta + O(1))\big).
\end{equation}
Because the resulting vector $\vec v_\epsilon$ is one of the individual vectors,
it has added height $3 \lg \Delta + O(1)$
and one additional square root.

Second, we compute $\vec v_{\epsilon/2}$ by
averaging $\vec v_\epsilon$ with $(1, 0)$,
negating the vector if the $y$ coordinate of $\vec v_\epsilon$ is negative,
and normalizing the resulting vector.
The normalization adds one more square root to the root set.%
\footnote{It is possible to avoid (and an earlier version of this paper avoided)
adding any roots, by approximating the (no longer unit)
vectors $\vec d_0, \vec d_1, \dots, \vec d_{k-1}$ so that the angle
between consecutive vectors $\vec d_i, \vec d_{i+1}$ is at most $\epsilon/2$
and $\Omega(\epsilon)$. Namely, for $\epsilon \leq 45^\circ$,
we can take $\vec d_i = (a,i)$ for the smallest integer $a$
for which the slope $1/a$ is less than the slope of~$\vec v_\epsilon$.}
Also, $\vec v_{\epsilon/2}$ has height only an additive constant
larger than $\vec v_\epsilon$.
Similarly, we can compute $\vec v_{\epsilon/4}$ from $\vec v_{\epsilon/2}$,
with a third added root.
The running time for this step is
\begin{equation} \label{eq:T2}
O\big(T'(P, 3 \lg \Delta + O(1))\big),
\end{equation}
where the $O(1)$ accounts for the three added roots
(one for normalizing each of $\vec v_\epsilon$, $\vec v_{\epsilon/2}$,
and $\vec v_{\epsilon/4}$).

Third, we construct $k = \left\lfloor \frac{2\pi}{\epsilon/4} \right\rfloor$
unit direction vectors $\vec d_0, \vec d_1, \dots, \vec d_{k-1}$
such that the counterclockwise angle between consecutive vectors
$\vec d_i, \vec d_{i+1}$ is exactly $\epsilon/4$,
and the wraparound pair $\vec d_{k-1}, \vec d_0$
has angle between $\epsilon/4$ and $\epsilon/2$.
For height efficiency, we use a recursive algorithm
(similar to the repeated squaring algorithm for computing a power).
As base cases, $\vec d_0 = (1,0)$ and $\vec d_1 = \vec v_{\epsilon/4}$
which are both unit vectors.
Then, for $i = 2, 3, \dots$, we define the $i$th unit direction vector $\vec d_i$
to have the sum of the angles of $\vec d_{\lfloor i/2 \rfloor}$ and
$\vec d_{\lceil i/2 \rceil}$, which we compute using cosine/sine sum formulas:
\begin{align*}
  \vec d_i = (x_i, y_i)
  &=\left(\cos i \tfrac{\epsilon}{4}, ~ \sin i \tfrac{\epsilon}{4}\right) \\
  &=\left(
      \cos \left[ \big\lfloor \tfrac{i}{2} \big\rfloor \tfrac{\epsilon}{4} + \big\lceil \tfrac{i}{2} \big\rceil \tfrac{\epsilon}{4} \right], ~
      \sin \left[ \big\lfloor \tfrac{i}{2} \big\rfloor \tfrac{\epsilon}{4} + \big\lceil \tfrac{i}{2} \big\rceil \tfrac{\epsilon}{4} \right]
    \right) \\
  &=\left(
      \cos \big\lfloor \tfrac{i}{2} \big\rfloor \tfrac{\epsilon}{4} \cos \big\lceil \tfrac{i}{2} \big\rceil \tfrac{\epsilon}{4} -
      \sin \big\lfloor \tfrac{i}{2} \big\rfloor \tfrac{\epsilon}{4} \sin \big\lceil \tfrac{i}{2} \big\rceil \tfrac{\epsilon}{4}, ~
      \sin \big\lfloor \tfrac{i}{2} \big\rfloor \tfrac{\epsilon}{4} \cos \big\lceil \tfrac{i}{2} \big\rceil \tfrac{\epsilon}{4} +
      \cos \big\lfloor \tfrac{i}{2} \big\rfloor \tfrac{\epsilon}{4} \sin \big\lceil \tfrac{i}{2} \big\rceil \tfrac{\epsilon}{4}
    \right) \\
  &=\left(
      x_{\lfloor i/2 \rfloor} x_{\lceil i/2 \rceil} - y_{\lfloor i/2 \rfloor} y_{\lceil i/2 \rceil}, ~
      y_{\lfloor i/2 \rfloor} x_{\lceil i/2 \rceil} + x_{\lfloor i/2 \rfloor} y_{\lceil i/2 \rceil}
    \right).
\end{align*}
The height of each $\vec d_i$ is at most
$2 \lg \lceil \frac{4}{\epsilon} \rceil = 2 \lg \frac{1}{\epsilon} + O(1)$
beyond the height of $\vec v_\epsilon$,
so a total added height of $2 \lg \frac{1}{\epsilon} + 3 \lg \Delta + O(1)$
relative to the inputs, and three added roots.
The running time of fully evaluating the recurrence is
\begin{equation} \label{eq:T3}
O\Big(\tfrac{1}{\epsilon} \cdot T'\big(P, 3 \lg \Delta + 2 \lg \tfrac{1}{\epsilon} + O(1)\big) \Big).
\end{equation}

Fourth, for each vertex $V$ of each face $F$ of~$P$,
we choose $O(1/\epsilon)$ direction vectors in $F$'s local coordinate system
by taking the two incident edges $e_1, e_2$ of~$F$ and the subset of the vectors
$\vec d_0, \vec d_1, \dots, \vec d_{k-1}$ that are within that wedge.
As in the real RAM algorithm, we check each $\vec d_i$ against each $e_j$
by comparing slopes and checking signs of coordinates to determine the quadrant.
These computations can be done via Operation~\ref{op:real}
with a running time of
\begin{equation*}
O\Big(\tfrac{\deg(V)}{\epsilon} \cdot T'\big(P, 3 \lg \Delta + 2 \lg \tfrac{1}{\epsilon} + O(1)\big) \Big)
\end{equation*}
for each vertex $V$, which by the Handshaking Lemma has a total of
\begin{equation} \label{eq:T4}
O\Big(\tfrac{n}{\epsilon} \cdot T'\big(P, 3 \lg \Delta + 2 \lg \tfrac{1}{\epsilon} + O(1)\big) \Big).
\end{equation}

Fifth, we show how to find an outgoing edge from any given node
$s = (U, A)$ in the graph~$G$, with corresponding vertex $U$ and
a wedge $A$ of directions from $\vec v_1$ to $\vec v_2$.
We follow the same construction in Theorem~\ref{FindQuasigeodesicTheorem},
but cone-following via Corollary~\ref{cor2:cone following}
instead of Corollary~\ref{cor:cone following}.
By Corollary~\ref{cor2:cone following}, following the cone
$(X,\vec d_i,\infty)$ and $(X,\vec d_{i+1},\infty)$
for $k = O\left( \frac{L^2}{\theta \, \ell^2} \right)
= O\left( \frac{L^2}{\epsilon \, \ell^2} \right)$ steps
costs
\begin{align}
  & O\Big(T'\big((P,\vec d_i,\vec d_{i+1}), 6 \lg k + O(1)\big) \, k \lg n\Big).
  \nonumber\\
  ={} & O\Big(T'\big(P, 2 \lg \tfrac{1}{\epsilon} + 3 \lg \Delta + 6 \lg k + O(1)\big) \, k \lg n\Big)
  \nonumber\\
  ={} & O\Big(T'\big(P, 2 \lg \tfrac{1}{\epsilon} + 3 \lg \Delta + 6 \lg \tfrac{L^2}{\epsilon \, \ell^2} + O(1)\big) \, \tfrac{L^2}{\epsilon \, \ell^2} \lg n\Big).
  \label{eq:follow}
\end{align}
Once we determine the reachable vertex $V$, we follow a similar construction
to Theorem~\ref{FindQuasigeodesicTheorem} to find the reachable angular range
$B$ of $V$.
We construct a planar embedding of the faces incident to $V$
in the same way as the first paragraph of this proof.
Then we do a binary search to choose the sector among
the $O(1/\epsilon)$ vectors around the embedding of~$V$.
The cost for this binary search is
\begin{equation} \label{eq:outgoing}
   O\Big(T'\big(P, 2 \lg \tfrac{1}{\epsilon} + 3 \lg \Delta + 6 \lg \tfrac{L^2}{\epsilon \, \ell^2} + O(1)\big) \lg \tfrac{1}{\epsilon} \Big).
\end{equation}
The total cost to find an outgoing edge is the sum of
expressions \eqref{eq:follow} and \eqref{eq:outgoing},
which is dominated by \eqref{eq:follow}.

Sixth, we traverse the graph in the same combinatorial way as
Theorem~\ref{FindQuasigeodesicTheorem}, but using the new
algorithm in the fifth paragraph to make each step.
As before, the number of steps is $O(n/\epsilon)$, so the
total running time is this number times expression \eqref{eq:follow}
(the running time of the subroutine in the fifth paragraph):
\begin{equation} \label{eq:T6}
   O\Big(T'\big(P, 2 \lg \tfrac{1}{\epsilon} + 3 \lg \Delta + 6 \lg \tfrac{L^2}{\epsilon \, \ell^2} + O(1)\big) \cdot \tfrac{n \, L^2}{\epsilon^2 \, \ell^2} \lg n \Big).
\end{equation}

The total running time on the word RAM is the sum of
expressions \eqref{eq:T1}, \eqref{eq:T2}, \eqref{eq:T3}, \eqref{eq:T4},
and \eqref{eq:T6}, of which \eqref{eq:T6} dominates.
Substituting Equation~\eqref{eq:T'(M,h)}, \eqref{eq:T6} expands to%
\begin{align*}
&
O\left( b^*(P)^2 \cdot \big(32 \cdot \max\{2,K(P)\}\big)^{2 |R(P)| + 2 H(P) + 4 \lg \tfrac{1}{\epsilon} + 6 \lg \Delta + 12 \lg \tfrac{L^2}{\epsilon \, \ell^2} + O(1)} \cdot \tfrac{L^2}{\epsilon^2 \, \ell^2} \, n \lg n \right).
\\
={}
&
O\left( b^*(P)^2 \cdot \big(32 \cdot \max\{2,K(P)\}\big)^{2 |R(P)| + 2 H(P) + 16 \lg \tfrac{1}{\epsilon} + 6 \lg \Delta + 24 \lg \tfrac{L}{\ell} + O(1)} \cdot \tfrac{L^2}{\epsilon^2 \, \ell^2} \, n \lg n \right).
\qedhere
\end{align*}
\end{proof}

\section{Conclusion}
\label{OpenQuestionsSection}

It has been known for seven decades \cite{DiscreteThreeGeodesics} that every
convex polyhedron has a closed quasigeodesic, but our algorithm is
the first finite algorithm to find one.
We end with some open problems about extending our approach,
though they all seem difficult.

\begin{open}
  Is there a reduction from sum-of-square-roots \cite{TOPP-33} to
  a decision problem involving quasigeodesics?
\end{open}

One goal would be to represent the sum of $n$ given square roots by
unfolding/following a geodesic ray across $\Theta(n)$ edges and faces,
where ideally the faces could come from a convex polyhedron.
As described in Section~\ref{sec:sum-of-square-roots}, such a reduction
would justify our algorithm taking exponential time in the worst case,
as no better bound is known for sum-of-square-roots.

\begin{open}
  Is there a pseudopolynomial-time algorithm on the real RAM
  for finding a \emph{non-self-intersecting} closed quasigeodesic?
  In particular, can we find the \emph{shortest} closed quasigeodesic?
\end{open}

At least three non-self-intersecting closed quasigeodesics exist
\cite{DiscreteThreeGeodesics}, but Theorem~\ref{FindQuasigeodesicTheorem}
does not necessarily find one.
Any approach similar to Theorem~\ref{FindQuasigeodesicTheorem} is unlikely to
resolve this, for several reasons:
\begin{enumerate}
\item Quasigeodesics could enter a vertex at a continuum of angles.
Theorem~\ref{FindQuasigeodesicTheorem} makes this manageable by grouping similar
angles of entry to a vertex, but if similar angles of entry to a vertex are
combined, extensions that would be valid for some of them but invalid for others
are treated as invalid for all of them. For instance, a quasigeodesic found by
Theorem~\ref{FindQuasigeodesicTheorem} will almost never turn by the maximum
allowed at any vertex, since exiting a vertex at the maximum possible turn from
one entry angle to the vertex may mean exiting it with more of a turn than
allowed for another very close entry angle. So there are some closed
quasigeodesics not findable by Theorem~\ref{FindQuasigeodesicTheorem}, and those
may include non-self-intersecting ones.
\item Given a vertex and a wedge determined by a range of directions from it, we
can find \emph{one} vertex in the wedge, but if we wish to find more than one,
the problem becomes more complicated. When we seek only one vertex, we only need consider
one unfolding of the faces, in which the entire wedge remains until it
hits a vertex; when we pass a vertex, the unfoldings on each side of it might be
different, so we multiply the size of the problem by 2 every time we pass a
vertex. There may, in fact, be exponentially many non-self-intersecting geodesic
paths between two vertices: for instance,
Aronov and O'Rourke \cite[Section~24.4]{Demaine-O'Rourke-2007}
give the
example of a doubly covered regular polygon, in which a geodesic path may visit
every vertex in order around the cycle but may skip vertices.
\end{enumerate}
%
Chartier and de Mesmay \cite{simple-geodesics}
give an algorithm to find a non-self-intersecting closed quasigeodesic
on a polyhedron, but the running time is pseudo-exponential
(exponential in both $n$ and $L/\ell$)
even on the real RAM; see Section~\ref{sec:simple-geodesics-time}.


\begin{open}
On the real RAM, the running time of Theorem~\ref{FindQuasigeodesicTheorem}
is polynomial in not just $n$ but the
smallest curvature at a vertex, the length of the longest edge, and the shortest
distance within a face between a vertex and an edge not containing it. Are all
of those necessary? Can the last be simplified to the length of the shortest
edge?
\end{open}

\begin{open}
  On the expression RAM (or word RAM),
  the running time of Theorem~\ref{FindQuasigeodesicTheorem}
  has an exponential dependence on $n$ and the geometric features.
  Can this be improved to a polynomial dependence,
  that is, a pseudopolynomial-time algorithm?
\end{open}



\begin{open}
Does every \emph{nonconvex} polyhedron (without boundary) have a closed quasigeodesic?
Can the algorithm of Theorem~\ref{FindQuasigeodesicTheorem} be extended to this case?
\end{open}

A quasigeodesic cannot pass through a nonconvex vertex (with more than $360^\circ$ of
material), as it is impossible to split into at most $180^\circ$ on both sides.
If the extended wedge in our algorithm contains a nonconvex vertex,
the wedge will split in two, as shown in Figure~\ref{fig:nonconvex}.
This could grow the complexity exponentially, and furthermore it is unclear how to
guarantee that the wedge eventually hits a convex vertex.

\begin{figure}
\centering
\includegraphics[width=0.7\linewidth]{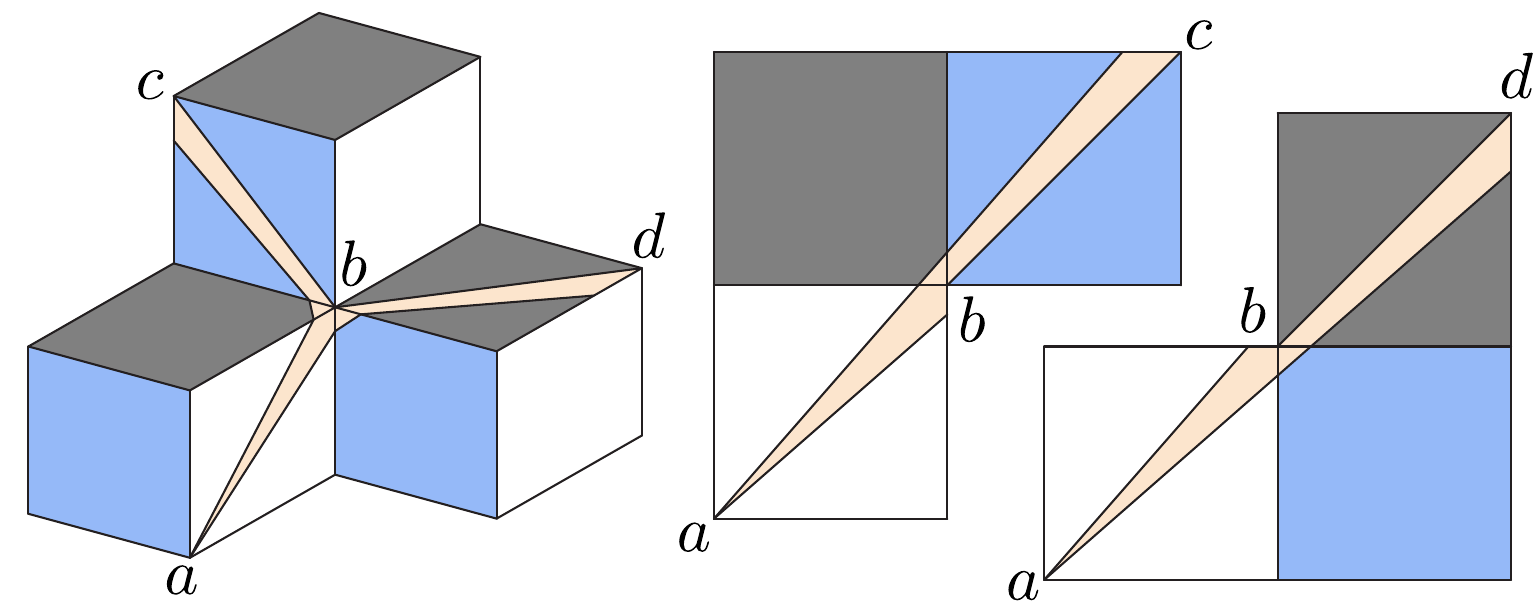}
\caption{An example of our algorithm applied to a nonconvex polycube. A
geodesic search from vertex $a$ within an angular interval may encounter a
nonconvex vertex from which the search space divides.}
\label{fig:nonconvex}
\end{figure}

\label{sec:nonconvex}
Chartier and de Mesmay \cite{simple-geodesics}
define an alternate notion of quasigeodesic path for nonconvex polyhedra,
which requires at negative-curvature vertices that the path has an angle of
\emph{at least} $180^\circ$ on both sides
(while at positive-curvature vertices the path
still has angles of at most~$180^\circ$).
They prove that every nonconvex polyhedron has such a closed ``quasigeodesic'',
and gave an algorithm to find one.
Because such paths exist, our algorithm can also find such a path,
redefining $\epsilon$ to be the smallest \emph{absolute}
curvature of a polyhedron vertex.
(Their algorithm additionally guarantees weak non-self-intersection,
at the cost of a larger running time.)

\section*{Acknowledgments}

We thank Zachary Abel, Nadia Benbernou, Fae Charlton, Jayson Lynch,
Joseph O'Rourke, Diane Souvaine, and David Stalfa for many discussions
related to this paper.
We also thank the referees for many helpful comments
that greatly improved and corrected this paper.

\bibliographystyle{alpha}
\bibliography{quasigeodesics}

\end{document}